\documentclass[lettersize,journal]{IEEEtran}

\usepackage{cite}
\usepackage{graphicx}
\usepackage{amsmath,amssymb,amsfonts,amsthm}
\usepackage{epstopdf}
\usepackage{flushend}
\usepackage{color}
\usepackage{soul}
\usepackage[caption=false,font=normalsize,labelfont=sf,textfont=sf]{subfig}

\newcommand{\newstuff}[1]{\textcolor{black}{#1}}

\newtheorem{theorem}{Theorem}

\newtheorem{corollary}{Corollary}
\newtheorem{remark}{Remark}
\newtheorem{proposition}{Proposition}

\allowdisplaybreaks
\IEEEoverridecommandlockouts

\begin{document}
\title{The Impact of Coherence Diversity on MIMO Relays}
\author{Fan~Zhang,~\IEEEmembership{Member,~IEEE} and Aria~Nosratinia,~\IEEEmembership{Fellow,~IEEE} 
\thanks{This work was made possible in part by the grant 1718551 from the National Science Foundation.} 
\thanks{The authors are with the Department of Electrical \& Computer Engineering, University of Texas at Dallas, Email: fan.zhang5@utdallas.edu; aria@utdallas.edu}}
\maketitle

\def\Wtilde{\widetilde{\W}}
\def\What{\hat{\W}}
\def\Xhat{\hat{\X}}
\def\Xtilde{\widetilde{\X}}
\def\Yhat{\hat{\Y}}
\def\Ytilde{\widetilde{\Y}}
\def\Hset{{\mathcal H}}
\def\Heq{\overline{\mathbf H}}
\def\EE{\mathbb{E}}
\def\A{{\mathbf A}}
\def\R{R}
\def\d{d}
\def\u{U}
\def\uset{\mathcal U}
\def\r{r}
\def\rdash{\hat{r}}
\def\q{q}
\def\bigO{O}
\def\ZeroMat{\mathbf 0}
\def\Onevec{\mathbf 1}
\def\XiNoise{\mathbf N}
\def\tX{{\mathbf X}}
\def\tx{{\mathbf x}}
\def\tH{{\textbf H}}
\def\tY{{\mathbf Y}}
\def\ty{{\mathbf y}}
\def\tW{{\mathbf W}}
\def\tw{{\mathbf w}}
\def\tI{{\mathbf I}}
\def\tA{{\mathbf A}}
\def\tR{{\mathbf R}}
\def\tC{{\mathbf C}}
\def\tD{{\mathbf D}}
\def\tB{{\mathbf B}}
\def\tU{{\mathbf U}}
\def\tG{{\mathbf G}}
\def\tV{{\mathbf V}}
\def\tT{{\mathbf T}}
\def\tu{{\mathbf u}}
\def\tv{{\mathbf v}}
\def\Mm{{\mathbf M}}

\graphicspath{{figures/}}

\makeatletter%
\if@twocolumn%
\newcommand{\Figwidth}{0.9\columnwidth}%
\newcommand{\hspaceonetwocol}[2]{\hspace{#2}}
\newcommand{\includeonetwocol}[2]{#2}
\def\twocolbreak{\nonumber\\ &}%
\def\twocolnewline{\nonumber\\}%
\def\twocolAlignMarker{&}%
\else
\newcommand{\Figwidth}{4.5in}%
\newcommand{\hspaceonetwocol}[2]{\hspace{#1}}
\newcommand{\includeonetwocol}[2]{#1}
\def\twocolbreak{}%
\def\twocolnewline{}%
\def\twocolAlignMarker{}%
\fi%
\makeatother%


\begin{abstract}
This paper studies MIMO relays with non-identical link coherence times, a frequently occurring condition when, e.g., the nodes in the relay channel do not all have the same mobility, or the scatterers around some nodes have different mobility compared with those around other nodes. Despite its practical relevance, this condition, known as {\em coherence diversity}, has not been studied in the relay channel. This paper studies the performance of MIMO relays and proposes efficient transmission strategies under coherence diversity. Since coherence times have a prominent impact on channel training, we do not assume channel state is available to the decoder for free; all channel training resources are accounted for in the calculations. A product superposition technique is employed at the source which allows a more efficient usage of degrees of freedom when the relay and the destination have different training requirements. Varying configurations of coherence times are studied. The interesting case where the different link coherence intervals are not a multiple of each other, and therefore the coherence intervals do not align, is studied. Relay scheduling is combined with the product superposition to obtain further gains in degrees of freedom. The impact of coherence diversity is further studied in the presence of multiple parallel relays.
\end{abstract}

\begin{IEEEkeywords}
Relay, cooperation, high mobility, fading, coherence time, channel training, channel state information, multiple relay
\end{IEEEkeywords}

\section{Introduction}
There has been a surge of interest in high-mobility wireless communications~\cite{zhou2015high,ghazal2016non,sun2014maximizing}, wherein the co-existence of low-mobility and high-mobility users  has been an accepted fact~\cite{chih1993microcell}. Naturally, faster nodes lead to links with shorter coherence intervals, and slower nodes experience links with longer coherence intervals. A deeper understanding of relay performance under unequal link coherence times provides new tools and techniques for high-mobility wireless communications. Relaying in high-mobility scenarios has been acknowledged as an important topic~\cite{zhang2020performance,chakraborty2017joint,khattabi2016improved}, but the implications of relaying under unequal link coherence intervals has been an open problem.

This paper studies the degrees of freedom of relaying under unequal coherence intervals for the three links within the relay channel, a condition that occurs commonly in practice, but its effect on the performance of relays has been unknown.\footnote{The early version of some of the results of the present paper appeared in~\cite{Fan_2019}.} The performance of fading relay channel under {\em equal} coherence intervals has been extensively studied~\cite{Krikidis_2012,Hu_2017,Jamali_2015,Shaqfeh_2016,Simoni_2016,Jamali_2014,wang2005capacity,fan2007mimo,bolcskei2006capacity}, and also the disparity of coherence intervals has been studied in broadcast channels~\cite{Li_2012,Li_2015,Fadel_disparity}, multiple access channels~\cite{Fadel_disparity}, and frequency-selective broadcast channels~\cite{Fadel_2019}. The impact of hybrid channel state information (CSI) on MISO broadcast channel with unequal coherence times was studied in~\cite{Fadel_2020}.

For the MIMO relay with coherence diversity,  we assume there is no free CSI at the receivers, since unequal coherence times impact channel training and therefore the assumption of free CSI would distort and obscure important features of the problem. In addition, no CSI is assumed at transmitters. We propose a product superposition at the source, a signaling strategy first introduced in~\cite{Li_2012}. Product superposition is a technique that allows a more efficient utilization of channel degrees of freedom when links with unequal coherence intervals coexist in a network. The gains of product superposition arise from a natural match between the product form of the superposition, and the product operation by which fading channels affect transmit signals. This match allows one signaling component to ``disappear'' into the effective channel gain of another link, and in that way reduce the interference. This principle is then employed to reuse pilot slots of one user for data of another user, without contaminating the pilots. As will become clear in the sequel, we use this technique to harvest gains when the source-destination and source-relay links have unequal coherence intervals.

A key novelty of the present paper is in its introduction of effective methods for controlling the interference arising from both the full-duplex relay and the source on all received pilots, while utilizing product superposition at the source to achieve transmission efficiency. This clearly separates the methods and analysis of the present study from previous works in the area of coherence diversity, e.g.,~\cite{Li_2012,Li_2015,Fadel_disparity,Fadel_2019}.

A summary of the contributions of this paper is as follows. We begin by proving the following: under identical coherence intervals for the source-relay, relay-destination, and source-destination links, the relay cannot provide any DoF gains compared with the direct link alone. This is a simple but important negative result that is independent of antenna configurations at the three nodes, and is used as a reference. When coherence intervals are unequal, we start with a representative example, design signaling appropriately for the unequal coherence intervals, and show the resulting DoF gains. Then we broaden the result by removing the constraints from the length and alignment of the coherence blocks, showing that the DoF gains persist in the more general case. Further, a new scheme combining the product superposition and relay scheduling is proposed, motivated by the following observation: Whenever a pilot-based relay is activated, the relay pilots impose a cost (in degrees of freedom) due to their interference with source-destination transmission. In the new scheme, this cost is compared against the relay gains, and the relay is activated accordingly. We show the extent to which this new scheme improves the degrees of freedom of the relay channel. This paper also studies multiple parallel relays under non-identical coherence intervals, wherein transmission strategies are studied and achievable degrees of freedom are calculated.

This paper is organized as follows. Section~\ref{Sec:Sys} presents the system model used in this paper, and in particular formalizes the way different coherence times are introduced into the analysis. In Section~\ref{sec:2user}, the foundations of the proposed techniques are demonstrated via an analysis of a single-relay system where, for simplicity, we assume that the three coherence intervals in the relay channel are pairwise divisible.  This simplifies the analysis by ensuring that an integer number of coherence intervals from each faster fading link fits into the coherence interval of each slower fading link, avoiding fractional overlaps between coherence intervals. Section~\ref{sec:scheduling} introduces the relay activation/deactivation scheduling in order to maximize the end-to-end DoF returns from the relay operation, as mentioned above. Section~\ref{sec:general} highlights techniques that generalize the earlier results of the paper to arbitrary coherence times (non-integer ratios). Section~\ref{sec:multiple} extends the results to multiple relays. Section~\ref{sec:conclusion} provides concluding remarks.

\textit{Notation}: Bold lower-case letters, e.g. $\tx$, denote column vectors. Bold upper-case letters, e.g. $\Mm$,  denote  matrices.
The Euclidean norm is denoted by $\|\tx\|$ and the Frobenius norm $\|\Mm\|_F$. 
The trace, conjugate,
transpose and conjugated transpose of $\Mm$ are denoted $tr(\Mm) $, 
$\Mm\*$, $\Mm^T$ and
$\Mm^H$, respectively. $\Mm^{-T} \triangleq (\Mm^{-1})^T$ and $\Mm^{-H} \triangleq (\Mm^{-1})^H$. $\tI_m$ and $\mathbf{0}_{m \times n}$ denote the $m \times m$ identity matrix and $m\times n$ zero matrix, respectively, and the dimensions are omitted if cleared from the context. \newstuff{The base of the logarithm throughout the paper is 2.}

\begin{figure}
\centering
\includegraphics[width=\Figwidth]{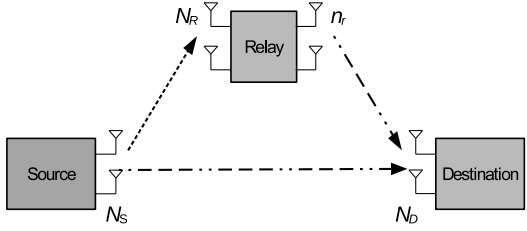}
\caption{MIMO Relay Channel with Coherence Diversity.}
\label{fig:RC_Strc}
\end{figure}

\section{System Model}
\label{Sec:Sys}
Consider a MIMO relay in full-duplex mode as in Figure~\ref{fig:RC_Strc}. The source and destination are equipped with $N_S$ and $N_D$ antennas, respectively. The relay has $N_R$ receive antennas and $n_R$ transmit antennas. The number of active (powered) relay transmit antennas in a transmission scheme is shown with $n_r$, which is optimized in each scenario. Obviously $n_r \leq n_R$.
The received signals at the relay and destination are:
\begin{align}
\ty_R & = \tH_{SR} \tx_S + \tw_{R} \\
\ty_D & = \tH_{SD} \tx_S + \tH_{RD} \tx_R + \tw_{D} ,
\label{Eq:Sys}
\end{align}
where $\tx_S$ and $\tx_R$ are signals transmitted from the source and relay. 
$\tw_R$ and $\tw_D$ are independent and identically distributed (i.i.d.) white Gaussian noise. The relay operates in the full-duplex mode. The relay self-interference residue can be modeled by a Gaussian distribution and can be additively combined with thermal noise into a single term $\tw_R$ for the purposes of analysis. This line of reasoning is by  now well-established for full duplex relay analysis, among others see \cite{DuarteThesis2012,Hong:CommMag14,8274994}. Bold upper case $\tH_{SR}$, $\tH_{RD}$ and $\tH_{SD}$ are channel gain matrices whose entries are i.i.d.\ Gaussian. Channel gain entries and noise components are zero-mean and are normalized to unit variance. Channel gains experience block fading, remaining constant during the coherence intervals which are, respectively, of length $T_{SR}$, $T_{RD}$ and $T_{SD}$, satisfying $T_{SR} \geq 2\max (N_S,N_R)$, $T_{RD} \geq 2\max (n_R,N_D)$ and $T_{SD} \geq 2\max (N_S,N_D)$. Channel gains are independent across blocks \cite{Zheng_2002}. The source and relay obey power constraints $\EE[ \text{tr}(\tx_S \tx_S')] \le \rho$ and $\EE[\text{tr}(\tx_R \tx_R')]  \le \rho$. We assume there is no CSIT at the source nor relay. \newstuff{If not stated, there is no free CSIR at the relay nor destination}.

The source sends messages to the destination with rate $R(\rho)$ at signal-to-noise ratio $\rho$. The achievable degrees of freedom at the destination achieving rate $R(\rho)$ are defined as
\begin{equation}
d = \lim_{\rho \rightarrow \infty} \frac{R(\rho)}{\log (\rho)}.
\end{equation}

\section{Aligned Coherence Blocks}
\label{sec:2user}
We first show the relay cannot provide any gains in degrees of freedom under {\em identical} coherence intervals \newstuff{and present the optimal degrees of freedom achieved by the direct link alone. We will use this optimal achievable degrees of freedom as a reference in the remaining of the paper}. Then we analyze the scenarios where the coherence times are unequal. \newstuff{Compared with the degrees of freedom achieved by the direct link alone, additional degrees of freedom can be obtained via product superposition, exploiting the coherence diversity between source-destination and source-relay links.}

\subsection{Identical Coherence Times}

\begin{proposition}
\label{Thm:Equal}
When relay link coherence times are identical ($T_{SD} = T_{SR} = T_{RD} = T$), the relay does not improve the degrees of freedom of the source-destination link, namely:
\begin{align}
d  =  \min (N_S,N_D) (1 - \frac{\min (N_S,N_D)}{T}).
\end{align}
\end{proposition}
\begin{proof}
From the cut-set bound, 
\begin{align}
R \leq \min \{\text{I}(\tx_S;\ty_R,\ty_D | \tx_R),\text{I}(\tx_S,\tx_R ; \ty_D)\}.
\end{align}

If $N_S \leq N_D$, consider the broadcast component of the cutset bound: $R \leq \text{I}(\tx_S;\ty_R,\ty_D | \tx_R)$. Because $T_{SD} = T_{SR} = T$ and there is no CSIT, the right hand side in the inequality is upper bounded by the capacity of a point-to-point channel having $N_S$ transmit antennas and $(N_D + N_R)$ receive antennas with coherence time $T$, which is $N_S(1 - \frac{N_S}{T}) \log \rho + o(\log \rho)$. Then we have
\begin{align}
d \leq N_S(1 - \frac{N_S}{T}),
\end{align}
which can be achieved by the direct link alone.

If $N_S \geq N_D$, we focus on the MAC component of the cutset bound: $R \leq \text{I}(\tx_S,\tx_R ; \ty_D)$. Since $T_{SD} = T_{SR} = T$, the right hand side  is upper bounded by the capacity of a point-to-point channel having $(N_S + n_R)$ transmit antennas and $N_D$ receive antennas with coherence time $T$, whose capacity is $N_D(1 - \frac{N_D}{T}) \log \rho + o(\log \rho)$. Then we have
\begin{align}
d \leq N_D(1 - \frac{N_D}{T}),
\end{align}
and this degrees of freedom can also be achieved by the direct link alone. This completes the proof.
\end{proof}

\subsection{A Representative Example for Unequal Coherence Times}
\label{sec:toy_unequal}
To pave the way for the analysis to come, and to motivate the direction taken by this paper, we provide an example whose purpose is to illuminate the main features of the problem in a simple setting.
In this example, the source and relay are equipped with two antennas and the destination is equipped with three antennas. The coherence times of the three links are as follows: $T_{SD} = T_{RD} = 8$ and $T_{SR} = \infty$, i.e., the source-relay channel is static, therefore the cost of training over this link is amortized over a large number of samples, so we can assume the relay knows $\tH_{SR}$. 

The source uses product superposition, sending
\begin{equation}
\tX_S = \tU [\tI_{2} , \mathbf{0}_{2 \times 1} , \tV],
\end{equation}
where $\tU \in \mathbb{C}^{2 \times 2}$ and $\tV \in \mathbb{C}^{2 \times 5}$.

At the relay, the received signal is
\begin{align}
\tY_R  = \tH_{SR} \tX_S + \tW_{R} 
       = \tH_{SR} \tU [\tI_{2} , \mathbf{0}_{2 \times 1} , \tV] + \tW_{R}.
\end{align}

The received signal during the first two time slots is
\begin{align}
\tY_R^{\prime}  =  \tH_{SR} \tU  +\tW_{R}^{\prime}.
\end{align}
The relay knows $\tH_{SR}$ and decodes $\tU$. The signal decoded by the relay in the previous block is $\tU^{\prime}$ and the two rows of $\tU^{\prime}$ are $\tu_1^{\prime}, \tu_2^{\prime} \in \mathbb{C}^{1 \times 2}$.

The relay powers only one antenna for transmission and sends
\begin{align}
\tX_R = [\mathbf{0}_{1 \times 2} , 1  , \tu_1^{\prime} , \tu_2^{\prime} , 0] \in \mathbb{C}^{1 \times 8}.
\end{align}

The received signal at the destination is:
\begin{align}
\nonumber
\tY_D & = \tH_{SD} \tX_S + \tH_{RD} \tX_R + \tW_{D} \\ \nonumber
	  & = [\tH_{SD},\tH_{RD}]  
	  \begin{bmatrix}
	  \tU [\tI_{2} , \mathbf{0}_{2 \times 1} , \tV] \\ \nonumber
	  \mathbf{0}_{1 \times 2} ,1  , \tu^{\prime}_1 , \tu^{\prime}_2 , 0
	  \end{bmatrix}
	  + \tW_D \\ 
	  & = [\tH_{SD}\tU,\tH_{RD}] \Big[\tI_{3 } , 
	  \begin{bmatrix}
      \tV \\
      \tu_1^{\prime} , \tu_2^{\prime} , 0
      \end{bmatrix}
	  \Big]
	  + \tW_D,
\end{align}
The destination estimates the equivalent channel $[\tH_{SD} \tU,\tH_{RD}]$ in the first three time slots and decodes $\tV, \tu_1^{\prime}$ and $\tu_2^{\prime}$. In this proposed scheme, the destination can achieve the degrees of freedom $(2 \times 5 + 2 \times 1 \times 2)/8 = 1.75$. In comparison, a traditional relaying scheme assigns pilots and training according to the smallest coherence time and achieves the degrees of freedom $2 \times (8 - 2)/8 = 1.5$.

\subsection{Coherence Conditions $T_{SR} = \infty$}

When $T_{SR} = \infty$, the training resources required for the source-relay link can be amortized over a long period and are therefore negligible. This scenario occurs when source and relay are either stationary, have a dominant line-of-sight component, or both.
\begin{theorem}
\label{Thm:TSRinfty}
In a relay channel with $T_{SR} = \infty$ and $N_S < N_D$, the following degrees of freedom are achievable, where $N_S^\ast \triangleq \min \{N_S,N_R\}$:

\noindent
If $T_{SD} = T_{RD}$,
\begin{align}
\label{dof:infty_SD=RD}
d = \max_{n_r} (1 & - \frac{N_S + n_r}{T_{SD}}) \nonumber \\ & \times \min \{N_S  + n_r, N_S + \frac{N_S^\ast N_S}{T_{SD} - n_r - N_S} \}.
\end{align}
If $T_{RD} = K T_{SD}$,
\begin{align}
\label{eq:infty2}
d  =  & \max_{n_r} (1 - \frac{N_S + n_r}{T_{SD}}) \nonumber \\
& \times \min \bigg\{(N_S + n_r)(1 +  \frac{(K-1)n_r}{K(T_{SD} - n_r - N_S)}), \nonumber \\
& \quad \quad \quad \quad \quad \quad \quad N_S(1 + \frac{K N_S^\ast + (K - 1)n_r}{K(T_{SD} - n_r -N_S)})\bigg\}.
\end{align}
If $T_{SD} = K T_{RD}$,
\begin{align}
\label{eq:infty3}
d = & \max_{n_r}(1 - \frac{n_r}{T_{RD}})  \nonumber \\
& \times \min \bigg\{(N_S + n_r)(1 + \frac{N_S}{K(T_{RD} - n_r)}), \nonumber \\
& \quad \quad \quad \quad \quad \quad \quad \quad \quad N_S(1 + \frac{K N_S^\ast - N_S}{K(T_{RD} - n_r)})\bigg\}.
\end{align}
\end{theorem}

\begin{figure*}
\centering
\includegraphics[width=4.5in]{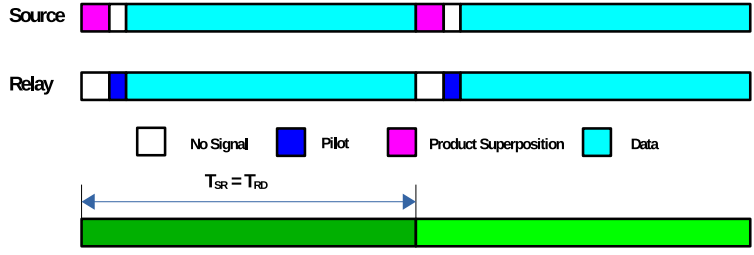}
\caption{Signaling Structure for Product Superposition.}
\label{fig:Sig_CohT}
\end{figure*}

\begin{proof}
When $T_{SD} = T_{RD}$, the proposed signaling structure is shown in Figure~\ref{fig:Sig_CohT}. In this case, the source sends the product superposition signal:
\begin{equation}
\tX_S = \tU [\tI_{N_S} , \mathbf{0}_{N_S \times n_r} , \tV_S],
\end{equation}
where $n_r \leq \min \{N_S, N_D-N_S\}$, $\tU \in \mathbb{C}^{N_S \times N_S}$ and $\tV_S \in \mathbb{C}^{N_S \times (T_{SD} - n_r - N_S)}$.

At the relay, the received signal is
\begin{align}
\tY_R  &= \tH_{SR} \tX_S + \tW_{R} \nonumber\\
      &= \tH_{SR} \tU [\tI_{N_S} , \mathbf{0}_{N_S \times n_r} , \tV_S] + \tW_R.
\end{align}

The received signal during the first $N_S$ time slots is
\begin{align}
\tY_R^{\prime}  =  \tH_{SR} \tU  +\tW_{R}^{\prime}.
\end{align}
The relay knows $\tH_{SR}$ and decodes $\tU$. Assume the message decoded by the relay in the previous block is $\tU^{\prime}$. The relay uses $n_r$ transmit antennas, sending
\begin{align}
\tX_R = [\mathbf{0}_{n_r \times N_S} , \tI_{n_r}  , \tV_R] \in \mathbb{C}^{n_r \times T_{SD}},
\end{align}
where $\tV_R \in \mathbb{C}^{n_r \times (T_{SD} - n_r - N_S)}$.

The received signal at the destination is
\begin{align}
\label{eq:infty-1-Yd}
\nonumber
\tY_D & = \tH_{SD} \tX_S + \tH_{RD} \tX_R + \tW_{D} \\ \nonumber
	  & = [\tH_{SD},\tH_{RD}]  
	  \begin{bmatrix}
	  \tU [\tI_{N_S} , \mathbf{0}_{N_S \times n_r} , \tV_S] \\ \nonumber
	  \mathbf{0}_{n_r \times N_S} , \tI_{n_r}  , \tV_R
	  \end{bmatrix}
	  + \tW_D \\
	  & = [\tH_{SD}\tU , \tH_{RD}] \Bigg[\tI_{(N_S + n_r)} ,
      \begin{bmatrix}
      \tV_S \\
      \tV_R
      \end{bmatrix}\Bigg]
	  + \tW_D.
\end{align}
The destination estimates the equivalent channel $[\tH_{SD}\tU, \tH_{RD}]$ during the first $(N_S + n_r)$ time slots and then decodes $\tV_S$ and $\tV_R$. At the destination, the decoded messages have two parts: $\tV_S$ from the source and $\tV_R$ from the relay, which provide degrees of freedom $N_S(T_{SD} - n_r - N_S)$ and $n_r(T_{SD} - n_r - N_S)$. The message in $\tV_R$ is from $\tU^{\prime}$. The degrees of freedom the relay can decode from $\tU^{\prime}$ are $N_S^\ast N_S$. The rate of the message emitted by the relay is bounded by the rate it decodes from the source. Taking the minimum of the degrees of freedom the relay can transmit and decode, and add it up with the degrees of freedom provided by the source, and optimizing the number of transmit antennas at the relay, the end-to-end degrees of freedom are \eqref{dof:infty_SD=RD}.

When $T_{RD} = K T_{SD}$, our scheme has a transmission block from the source that has length $KT_{SD}$, which we divide into sub-blocks of length $T_{SD}$. 
During the first sub-block, the source sends the signal
\begin{equation}
\label{sig:TH1C2XS1}
\tX_S^1 = \tU^1 [\tI_{N_S} , \mathbf{0}_{N_S \times n_r} , \tV_S^1],
\end{equation}
where $n_r \leq \min \{N_S, N_D-N_S\}$, $\tU^1 \in \mathbb{C}^{N_S \times N_S}$ and $\tV_S^1 \in \mathbb{C}^{N_S \times (T_{SD} - n_r - N_S)}$.

The relay decodes $\tU^1$ and uses $n_r$ transmit antennas and sends
\begin{align}
\label{sig:TH1C2XR1}
\tX_R^1 = [\mathbf{0}_{n_r \times N_S} ~ \tI_{n_r }  ~ \tV_R^1] \in \mathbb{C}^{n_r \times T_{SD}},
\end{align}
where $\tV_R^1 \in \mathbb{C}^{n_r \times (T_{SD} - n_r - N_S)}$. The received signal at the destination is
\begin{align}
\tY_D^1 \twocolAlignMarker = \tH_{SD}^1 \tX_S^1 + \tH_{RD} \tX_R^1 + \tW_{D}^1 \twocolnewline{}
	  \twocolAlignMarker = [\tH_{SD}^1\tU^1,\tH_{RD}] \Big[\tI_{(N_S + n_r)} ,
	  \begin{bmatrix}
      \tV_S^1 \\
      \tV_R^1
      \end{bmatrix}
	  \Big]
	  + \tW_D^1.
\end{align}
In the first sub-block, the three signal components $\tV_S^1$, $\tV_R$ and $\tU^1$ respectively provide for the degrees of freedom $N_S(T - n_r - N_S)$, $n_r(T - n_r - N_S)$ and $N_S^\ast N_S$. 

During the following $(K - 1)$ sub-blocks, the source sends the signal
\begin{equation}
\label{sig:TH1C2XS2}
\tX_S^k = \tU^k [\tI_{N_S} , \tV_S^k], \ 2 \leq k \leq K,
\end{equation}
where $\tV_S^k \in \mathbb{C}^{N_S \times (T_{SD} - n_r - N_S)}$.

The relay uses $n_r$ transmit antennas and sends:
\begin{align}
\label{sig:TH1C2XR2}
\tX_R^k = [\mathbf{0}_{n_r \times N_S} , \tV_R^k] \in \mathbb{C}^{n_r \times T_{SD}},
\end{align}
where $\tV_R^k \in \mathbb{C}^{n_r \times (T_{SD} - N_S)}$. The received signal at the destination is
\begin{align}
\tY_D^k \twocolAlignMarker = \tH_{SD}^k \tX_S^k + \tH_{RD} \tX_R^k + \tW_{D}^k \twocolnewline
	  \twocolAlignMarker = [\tH_{SD}^k\tU^k, \tH_{RD}] \Big[\tI_{N_S} ,
	  \begin{bmatrix}
      \tV_S^k \\
      \tV_R^k
      \end{bmatrix}
	  \Big]
	  + \tW_D^k.
\end{align}

During Sub-block $k$, the destination can decode $\tV_S^k$, $\tV_R^k$ and $\tU^k$, which respectively provide degrees of freedom $N_S(T - N_S)$, $n_r(T - N_S)$ and $N_S^\ast N_S$. Therefore, the end-to-end degrees of freedom are \eqref{eq:infty2}.

When $T_{SD} = K T_{RD}$, our source transmission block has length $KT_{RD}$, with sub-blocks of length $T_{RD}$. For the first sub-block, the source uses product superposition, sending
\begin{equation}
\tX_S^1 = \tU [\tI_{N_S} , \mathbf{0}_{N_S \times n_r} , \tV_S^1] \in \mathbb{C}^{N_S \times T_{RD}}.
\end{equation}
In the remaining $K-1$ sub-blocks with length $T_{RD}$, the source sends
\begin{equation}
\tX_S^k = \tU [\mathbf{0}_{N_S \times n_r} , \tV_S^k] \in \mathbb{C}^{N_S \times T_{RD}},
\end{equation}
where $n_r \leq \min \{N_S, N_D-N_S\}$, $\tU \in \mathbb{C}^{N_S \times N_S}$, $\tV_S^1 \in \mathbb{C}^{N_S \times (T_{RD} - n_r - N_s)}$ and $\tV_S^k \in \mathbb{C}^{N_S \times (T_{RD} - n_r)}$, $k = 2,3,\dots,K$.

The received signal during the first $N_S$ time slots is
\begin{align}
\tY_R^{\prime}  =  \tH_{SR} \tU  +\tW_{R}^{\prime}.
\end{align}
The relay knows $\tH_{SR}$ and decodes $\tU$. Then it uses $n_r$ transmit antennas and sends
\begin{align}
\tX_R^1 = [\mathbf{0}_{n_r \times N_S} , \tI_{n_r}  , \tV_R^1] \in \mathbb{C}^{n_r \times T_{RD}},
\end{align}
during the first sub-block with length $T_{RD}$. In the remaining $K-1$ sub-block the relay sends 
\begin{align}
\tX_R^k = [\tI_{n_r}  , \tV_R^k] \in \mathbb{C}^{n_r \times T_{RD}}.
\end{align}

During the first sub-block, the received signal at the destination is
\begin{align}
\nonumber
\tY_D^1 & = \tH_{SD} \tX_S^1 + \tH_{RD}^1 \tX_R^1 + \tW_{D}^1 \\ \nonumber
	  & = [\tH_{SD},\tH_{RD}^1]  
	  \begin{bmatrix}
	  \tU [\tI_{N_S } , \mathbf{0}_{N_S \times n_r} , \tV_S^1] \\ 
	  \mathbf{0}_{n_r \times N_S} , \tI_{n_r}  , \tV_R^1
	  \end{bmatrix}
	  + \tW_D^1 \\
	  & = [\tH_{SD}\tU,\tH_{RD}^1] \Bigg[\tI_{(N_S + n_r)} , 
	  \begin{bmatrix}
      \tV_S^1 \\
      \tV_D^1
      \end{bmatrix}
	  \Bigg]
	  + \tW_D^1.
\end{align}

The destination estimates the equivalent channel $[\tH_{SD}\tU, \tH_{RD}^1]$ during the first $(N_S + n_r)$ time slots and decodes $\tV_S^1$. 

During Sub-block $k$, the received signal at the destination is
\begin{align}
\nonumber
\tY_D^k & = \tH_{SD} \tX_S^k + \tH_{RD}^k \tX_R^k + \tW_{D}^k \\ \nonumber
	  & = [\tH_{SD},\tH_{RD}^k]  
	  \begin{bmatrix}
	  \tU [\mathbf{0}_{N_S \times n_r} , \tV_S^k] \\ \nonumber
	  \tI_{n_r }  , \tV_R^k
	  \end{bmatrix}
	  + \tW_D^k \\
	  & = \Bigg[\tH_{RD}^k , [\tH_{SD}\tU,\tH_{RD}^k]
	  \begin{bmatrix}
      \tV_S^k \\
      \tV_R^k
      \end{bmatrix}
	  \Bigg]
	  + \tW_D^k.
\end{align}

The destination estimates $\tH_{RD}^k$ during the first $n_r$ time slots. Because the destination already estimated $\tH_{SD}\tU$, it knows the equivalent channel $[\tH_{SD}\tU,~\tH_{RD}^k]$. During the remaining time slots, the destination decodes $\tV_S^k$, $\tV_R^k$, which respectively provide degrees of freedom $N_S(T - n_r -N_s)$ and $N_S(T - n_r)$. The data matrix $\tV_R^1$ provides degrees of freedom $n_r(T - n_r - N_s)$ and $\tV_R^k (2 \leq k \leq K)$ provides degrees of freedom $n_r(T - n_r)$. Adding up the degrees of freedom and optimizing the number of transmit antennas at the relay produces \eqref{eq:infty3}.
This completes the proof.
\end{proof}

\begin{corollary}
The degrees of freedom in Theorem~\ref{Thm:TSRinfty} are optimal under channel conditions $T_{SD}=T_{SR}$ and antenna configuration:
\begin{align}
\label{eq:optimal_thm2}
    &(N_D^{\ast} - N_S)(T_{SD} - N_D^{\ast}) \leq N_S^\ast N_S
\end{align}
where $N_D^{\ast} \triangleq \min \{N_S + n_R, N_D \}$. In this case, the DoF is: 
\begin{align}
    d_{opt} = N_D^\ast(1 - \frac{N_D^\ast}{T_{SD}}).
\end{align}
\end{corollary}

\begin{proof}
For achievability, the relay activates $n_r = N_D^\ast - N_S$ antennas for transmission. Because the condition \eqref{eq:optimal_thm2} holds (equivalent to $n_r \leq \frac{N_S^\ast N_S}{T_SD - n_r - N_S}$), according to Theorem~\ref{Thm:TSRinfty}, the degrees of freedom 
\begin{align*}
     \frac{1}{T_{SD}}\{N_S(T_{SD} - n_r - N_S) + & n_r(T_{SD} - n_r - N_S)\} \nonumber \\
    & = N_D^\ast(1 - \frac{N_D^\ast}{T_{SD}})
\end{align*}
are achievable.
For the converse, from the cut-set bound, the capacity of the relay is upper bounded by $I(\tY_D;\tX_R,\tX_S)$. Because the coherence times of the source-destination and relay-destination links are identical and the coherence blocks are aligned, this mutual information is equivalent to the capacity of a point-to-point channel with $N_S + n_R$ transmit antennas and $N_D$ receive antennas with coherence time $T_{SD}$. The degrees of freedom upper bound for this point-to-point channel is $N_D^\ast(1 - \frac{N_D^\ast}{T_{SD}})$. This completes the proof.
\end{proof}

\begin{corollary}
When $N_S < N_D, T_{SD} = T_{RD}$ and $T_{SR}=\infty$, the achievable degrees of freedom achieved by product superposition are strictly greater than that of source-destination link alone.
\end{corollary}
\begin{proof}
From Theorem~\ref{Thm:TSRinfty}, the direct link alone can achieve the following degrees of freedom: $d^{\prime} = \frac{N_S}{T_{SD}} \times (T_{SD} - N_S)$. Choose $n_r = 1$. If $n_r > \frac{N_S^\ast N_S}{T_{SD} - n_r - N_S}$, the degrees of freedom achieved by the proposed scheme are
\begin{align}
d & \geq \frac{1}{T_{SD}} (N_S(T_{SD} - 1 - N_S) + N_S N_S^{\ast}) \nonumber \\
& = \frac{N_S}{T_{SD}}(T_{SD} - N_S + N_S^\ast-1).
\end{align}
Obviously, $d \geq d^{\prime}$;
if $n_r \leq \frac{N_S^\ast N_S}{T_{SD} - n_r - N_S}$, the degrees of freedom achieved are
\begin{align}
d \twocolAlignMarker \geq \frac{1}{T_{SD}} (N_S(T_{SD} - 1 - N_S) + (T_{SD} - 1 - N_S)) \twocolnewline{}
     \twocolAlignMarker = \frac{N_S + 1}{T_{SD}}(T_{SD} - 1 - N_S).
\end{align}
Because $T_{SD} \geq 2N_D \geq 2N_S + 2$, 
\begin{align}
d & \geq   \frac{N_S + 1}{T_{SD}}(T_{SD} - 1 - N_S)  \nonumber \\
             & > \frac{N_S}{T_{SD}}(T_{SD} - N_S) = d^{\prime}.
\end{align}
This completes the proof.
\end{proof}

\begin{remark}
Theorem~\ref{Thm:TSRinfty} highlights strictly positive gains. But under some conditions, e.g. when the relay-destination coherence time is too short, the relay pilot requirements will eat into the gains. For example, when $N_S = N_R =3,N_D = 5, T = 4, K = 3$, $d = 2, d^{\prime} = \frac{9}{4}$, the relay does not provide any DoF gains. In this calculation we assumed $n_r>0$, i.e., the relay is not inactive.
\end{remark}

\begin{remark}
When $T_{SR}=T_{SD}=\infty$, the direct link can achieve the DoF outer bound, which is $\min (N_S,N_D)$. When $T_{SR}=T_{RD}=\infty$, set $K$ in \eqref{eq:infty2} to $\infty$, the DoF is $\max_{n_r} \min \big\{(N_S + n_r)(1 - \frac{N_S}{T_{SD}}), N_S (1 - \frac{N_S - N_S^\ast}{T_{SD}}) \big\}$.
\end{remark}

\begin{remark}
In the forward decoding version of block Markov encoding, used in our paper, the relay does not transmit data in the first block. This lack of transmission amortizes to zero over a long sequence of blocks, thus need not be included in the rate or DoF analysis.
\end{remark}

\subsection{Coherence Conditions $T_{SR} < \infty$}
When $T_{SR}$ is bounded, one can no longer assume that the relay knows $\tH_{SR}$ with negligible training cost. The following theorem states the achievable degrees of freedom. 
\begin{theorem}
\label{thm:int}
In a relay channel with link coherence times $T_{SR} = K T_{SD}$, and antenna configuration $N_S < N_D$, the following degrees of freedom are achievable:

\noindent
If $T_{SD} = T_{RD}$,
\begin{align}
\label{eq:Th2D1}
d = &  \max_{n_r} (1 - \frac{N_S + n_r}{T_{SD}}) \nonumber \\
& \times \min \{N_S + n_r,N_S + \frac{(K - 1)N_S^\ast N_S}{K(T_{SD} - n_r - N_S)} \}.
\end{align}
If $T_{RD} = K' T_{SD}$ and all coherence length pairs have integer ratios, equivalently  $\frac{\max(K,K')}{\min(K,K')}\in \mathbb N$,
\begin{align}
\label{eq:TH2D2}
d =  & \max_{n_r} (1 - \frac{N_S + n_r}{T_{SD}}) \twocolAlignMarker \nonumber \\
&\times \min \Big\{(N_S + n_r)(1 + \frac{(K'-1)n_r}{K'(T_{SD} - n_r - N_S)}), \nonumber \\
& \quad \quad  N_S(1 + \frac{K'(K - 1)N_S^\ast + K(K' - 1)n_r}{KK'(T_{SD} - n_r -N_S)})\Big\}.
\end{align}
If $T_{SD} = K' T_{RD}$,
\begin{align}
\label{eq:TH2D3}
d = & \max_{n_r}(1 - \frac{n_r}{T_{RD}}) \nonumber\\
& \times \min \Big\{(N_S + n_r)(1 - \frac{N_S}{K'(T_{RD} - n_r)}),\nonumber \\
&\quad \quad \quad \quad \quad \quad N_S(1 + \frac{(K - 1) N_S^\ast - K N_S}{K K'(T_{RD} - n_r)})\Big\}.
\end{align}
\end{theorem}

\begin{proof}
 When $T_{SD} = T_{RD}$, our transmission block has length $KT_{SD}$. This transmit block has $K$ sub-blocks with length $T_{SD}$. During the first sub-block, the source sends the signal
\begin{equation}
\tX_S^{1} = [\tI_{N_S } , \mathbf{0}_{N_S \times n_r} , \tV_S^1],
\end{equation}
where $\tV^{1} \in \mathbb{C}^{N_S \times (T_{SD} - n_r - N_S)}$. The destination estimates the channel $\tH_{SD}$. The relay estimates $\tH_{SR}$ during the first $N_S$ time slots.

In the next $(K - 1)$ sub-blocks, the source sends
\begin{equation}
\tX_S^k = \tU^k [\tI_{N_S } , \mathbf{0}_{N_S \times n_r} , \tV_S^k], \ k=2,\dots,K,
\end{equation}
where $n_r \leq \min \{N_S, N_D-N_S\}$, $\tU^k \in \mathbb{C}^{N_S \times N_S}$ and $\tV_S^k \in \mathbb{C}^{N_S \times (T_{SD} - n_r - N_S)}$.

The received signal at the relay is
\begin{align}
\tY_{R}^{k} = \tH_{SR} \tU^k [\tI_{N_S } , \mathbf{0}_{N_S \times n_r} , \tV_S^k].
\end{align}
The relay knows $\tH_{SR}$ and decodes $\tU^k$ and uses $n_r$ transmit antennas, sending
\begin{align}
\tX_R^k = [\mathbf{0}_{n_r \times N_S} , \tI_{n_r}  , \tV_R^k] \in \mathbb{C}^{n_r \times T_{SD}},
\end{align}
where $\tV_R^k \in \mathbb{C}^{n_r \times (T_{SD} - n_r - N_S)}$. The received signal at the destination is:
\begin{align}
\tY_D^k \twocolAlignMarker = \tH_{SD}^k \tX_S^k + \tH_{RD}^k \tX_R^k + \tW_{D}^k \twocolnewline
	  \twocolAlignMarker = [\tH_{SD}^k \tU^k,\tH_{RD}^k] \Big[\tI_{(N_S + n_r)},
	  \begin{bmatrix}
      \tV_S^k \\
      \tV_R^k
      \end{bmatrix}
	  \Big]
	  + \tW_D^k.
\end{align}
The destination therefore estimates the equivalent channel $[\tH_{SD}^k\tU^k,\tH_{RD}^k]$ during the first $(N_S + n_r)$ time slots, and decodes $\tV_S^k$ and $\tV_R^k$, respectively, provide degrees of freedom $N_S(T - n_r - N_S)$ and $n_r(T - n_r - N_S)$ per transmit block of length $T_{SD}$. For all $k\in\{2,\ldots,K\}$, the degrees of freedom provided via $\tU^k$ are $N_S^\ast N_S$, hence the total degrees of freedom the relay can decode are $(K - 1)N_S^\ast N_S$ per transmit block of length $K T_{SD}$. Therefore, adding up the degrees of freedom during the super block of length $T_{SD}$ and optimizing the number of relay transmit antennas, the end-to-end degrees of freedom are given by \eqref{eq:Th2D1}. This completes the first part of the theorem.

We now consider $T_{RD} = K' T_{SD}$. Recall that in this section we are focusing on fading blocks that are aligned, thus the ratio of any pair of coherence times is an integer. Therefore we have the following two cases for the coherence time configurations.

In the first case, the coherence time $T_{RD} = (K'/K)T_{SR} = K' T_{SD}$, where $(K'/K)$ is an integer. Our transmission block from the source has length $T_{RD}$ and is divided into sub-blocks with length $T_{SR}$. During the first sub-block, from time slot $1$ to $T_{SD}$, the source sends the signal
\begin{equation}
\tX_S^1 = [\tI_{N_S} , \mathbf{0}_{N_S \times n_r}, \tV_S^1] \in \mathbb{C}^{N_S \times T_{SD}},
\end{equation}
where $n_r \leq \min \{N_S, N_D-N_S\}$, $\tV^1 \in \mathbb{C}^{N_S \times (T_{SD} - n_r - N_S)}$.

The relay estimates $\tH_{SR}$ and sends
\begin{align}
\tX_R^1 = [\mathbf{0}_{n_r \times N_S} , \tI_{n_r }  , \tV_R^1] \in \mathbb{C}^{n_r \times T_{SD}},
\end{align}
where $\tV_R^1 \in \mathbb{C}^{n_r \times (T_{SD} - n_r - N_S)}$. The received signal at the destination is
\begin{align}
\tY_D^1 = [\tH_{SD}^1,~\tH_{RD}] \Big[\tI_{(N_S + n_r)}, ~ 
	  \begin{bmatrix}
      \tV_S^1 \\
      \tV_R^1
      \end{bmatrix}
	  \Big]
	  + \tW_D^1.
\end{align}
The destination estimates $[\tH_{SD}^1,~\tH_{RD}]$ and decodes $\tV_S^1$ and $\tV_R^1$. Then every $T_{SD}$ time slots, the source sends the signal
\begin{equation}
\tX_S^k = \tU^{k-1}[\tI_{N_S} , \mathbf{0}_{N_S \times n_r} , \tV_S^k] \in \mathbb{C}^{N_S \times T_{SD}} ,
\end{equation}
where $\tU^{k-1} \in \mathbb{C}^{N_S \times N_S}$, $\tV_S^k \in \mathbb{C}^{N_S \times (T_{SD} - n_r - N_S)}$. The relay decodes $\tU^{k-1}$ and sends
\begin{align}
\label{sig:TH2C1XR1}
\tX_R^k = [\mathbf{0}_{n_r \times N_S} , \tV_R^k] \in \mathbb{C}^{n_r \times T_{SD}},
\end{align}
where $\tV_R^k \in \mathbb{C}^{n_r \times (T_{SD} - N_S)}$. The received signal at the destination is
\begin{align}
\tY_D^k  = [\tH_{SD}^k\tU^k,\tH_{RD}] \Bigg[\tI_{N_S}, 
	  \begin{bmatrix}
      \tV_S^k \\
      \tV_R^k
      \end{bmatrix}
	  \Bigg]
	  + \tW_D^k.
\end{align}
The destination can decode $\tV_S^k$ and  $\tV_R^k$ which respectively provide degrees of freedom $N_S(T - N_S)$ and $n_r(T - N_S)$ and $\tU^k$ can provide degrees of freedom $N_S^\ast N_S$. 

In the remaining sub-block of length $T_{SR}$, the relay-destination channel keeps constant and it has already been estimated by the destination. Therefore, the relay does not need to send pilots. Then every $T_{SD}$ time slots, the transmitted signals at the source are:\footnote{The following expression represents the signaling structure, the information carrying matrices $\tU^k$ and $\tV_S^k$ are independent across different sub-blocks, but for convenience, we use the same notation across different sub-blocks.}
\begin{equation}
\tX_S^1 = [\tI_{N_S} , \tV_S^k],
\end{equation}
\begin{equation}
\tX_S^k = \tU^k[\tI_{N_S} ,  \tV_S^k], 2 \leq k \leq K,
\end{equation}
where $n_r \leq \min \{N_S, N_D-N_S\}$, $\tV_S^k \in \mathbb{C}^{N_S \times (T_{SD} - n_r - N_S)}$. The relay decodes $\tU^k$ and sends \eqref{sig:TH2C1XR1}, with the codeword representing the message of the latest decoded $\tU^k$. In this way, the destination can decode $\tV_S^k$ and $\tV_R^k$. 

The degrees of freedom the relay can decode are $\frac{K'}{K}(K - 1)N_S^\ast N_S$. The source-destination link achieves total degrees of freedom $N_S(K'T_{SD} - n_r -K' N_S)$. The relay-destination link achieves total degrees of freedom $n_r(K'T_{SD} - n_r -K' N_S)$. Adding it up with the degrees of freedom the source-destination link achieves and optimizing the number of relay transmit antennas, it results in the achievable degrees of freedom in \eqref{eq:TH2D2}.

In the second case, the coherence time $T_{SR} = (K/K')T_{RD} = K T_{SD}$, where $(K/K')$ is an integer. The transmission block from the source has length $T_{SR}$ and is divided into sub-blocks with length $T_{RD}$. During the first sub-block, from time slot $1$ to $T_{SD}$, the source sends the signal
\begin{equation}
\tX_S^1 = [\tI_{N_S} , \mathbf{0}_{N_S \times n_r} , \tV_S^1] \in \mathbb{C}^{N_S \times T_{SD}},
\end{equation}
where $n_r \leq \min \{N_S, N_D-N_S\}$, $\tV^1 \in \mathbb{C}^{N_S \times (T_{SD} - n_r - N_S)}$.

The relay estimates $\tH_{SR}$ and sends
\begin{align}
\tX_R^1 = [\mathbf{0}_{n_r \times N_S} , \tI_{n_r }  , \tV_R^1] \in \mathbb{C}^{n_r \times T_{SD}},
\end{align}
where $\tV_R^1 \in \mathbb{C}^{n_r \times (T_{SD} - n_r - N_S)}$. The received signal at the destination is
\begin{align}
\tY_D^1 = [\tH_{SD}^1, \tH_{RD}] \Bigg[\tI_{(N_S + n_r)}, ~ 
	  \begin{bmatrix}
      \tV_S^1 \\
      \tV_R^1
      \end{bmatrix}
	  \Bigg]
	  + \tW_D^1.
\end{align}
The destination estimates $[\tH_{SD}^1,\tH_{RD}]$ and decodes $\tV_S^1$ and $\tV_R^1$. Then every $T_{SD}$ time slots, the source sends the signal
\begin{equation}
\label{sig:TH2C2XS2}
\tX_S^k = \tU^{k-1}[\tI_{N_S}  , \tV_S^k] \in \mathbb{C}^{N_S \times T_{SD}},
\end{equation}
where $\tU^{k-1} \in \mathbb{C}^{N_S \times N_S}$, $\tV_S^k \in \mathbb{C}^{N_S \times (T_{SD} - N_S)}$. The relay decodes $\tU^{k-1}$ and sends
\begin{align}
\label{sig:TH2C2XR2}
\tX_R^k = [\mathbf{0}_{n_r \times N_S} , \tV_R^k] \in \mathbb{C}^{n_r \times T_{SD}},
\end{align}
where $\tV_R^k \in \mathbb{C}^{n_r \times (T_{SD} - N_S)}$. The received signal at the destination is
\begin{align}
\tY_D^k  = [\tH_{SD}^k\tU^k,\tH_{RD}] \Bigg[\tI_{N_S} , 
	  \begin{bmatrix}
      \tV_S^k \\
      \tV_R^k
      \end{bmatrix}
	  \Bigg]
	  + \tW_D^k.
\end{align}

During each remaining sub-block of length $T_{RD}$, from time slot $1$ to $T_{SD}$, the source sends the signal
\begin{equation}
\tX_S^1 = \tU^{1}[\tI_{N_S} , \mathbf{0}_{N_S \times n_r} , \tV_S^1],
\end{equation}
where $n_r \leq \min \{N_S, N_D-N_S\}$, $\tV^1 \in \mathbb{C}^{N_S \times (T_{SD} - n_r - N_S)}$.

The relay decodes $\tU^1$ and sends
\begin{align}
\tX_R^1 = [\mathbf{0}_{n_r \times N_S} , \tI_{n_r }  , \tV_R^1] \in \mathbb{C}^{n_r \times T_{SD}},
\end{align}
where $\tV_R^1 \in \mathbb{C}^{n_r \times (T_{SD} - n_r - N_S)}$. The received signal at the destination is
\begin{align}
\tY_D^1 = [\tH_{SD}^1,\tH_{RD}] \Bigg[\tI_{(N_S + n_r)},  
	  \begin{bmatrix}
      \tV_S^1 \\
      \tV_R^1
      \end{bmatrix}
	  \Bigg]
	  + \tW_D^1.
\end{align}
The destination estimates $[\tH_{SD}^1,\tH_{RD}]$ and decodes $\tV_S^1$ and $\tV_R^1$. Then every $T_{SD}$ time slots, the source and the relay sends the signal with the same structure as \textcolor{black}{\eqref{sig:TH2C2XS2},\eqref{sig:TH2C2XR2}.} 
The received signal at the destination is
\begin{align}
\tY_D^k  = [\tH_{SD}^k\tU^k,\tH_{RD}] \Bigg[\tI_{N_S} , 
	  \begin{bmatrix}
      \tV_S^k \\
      \tV_R^k
      \end{bmatrix}
	  \Bigg]
	  + \tW_D^k.
\end{align}

The degrees of freedom the relay can decode are $(K - 1)N_S^\ast N_S$. The source-destination link achieves total degrees of freedom $N_S(KT_{SD} - \frac{K}{K'}n_r -K N_S)$. The relay-destination link can provide total degrees of freedom $n_r(KT_{SD} - \frac{K}{K'}n_r -K N_S)$. Take the minimum of the degrees of freedom the relay can decode and can transmit. Adding up with the degrees of freedom the source-destination link achieves and optimizing the number of relay transmit antennas, it results in the achievable degrees of freedom in \eqref{eq:TH2D2}. This completes the second part of the theorem.

When $T_{SD} = K' T_{RD}$, our source transmission block has length $T_{SR}$ and is divided into sub-blocks with length $T_{SD}$. In the first sub-block, during time slot 1 to $T_{RD}$, the source sends
\begin{equation}
\tX_S^1 = [\tI_{N_S} , \mathbf{0}_{N_S \times n_r} , \tV_S^1] \in \mathbb{C}^{N_S \times T_{RD}}.
\end{equation}

The relay estimates $\tH_{SR}$ and sends 
\begin{align}
\tX_R^1 = [\mathbf{0}_{n_r \times N_S},\tI_{n_r}  , \tV_R^k] \in \mathbb{C}^{n_r \times T_{RD}}.
\end{align}

During the remaining $(K' - 1)T_{RD}$ time slots, every $T_{RD}$ time slots, the source sends
\begin{equation}
\tX_S^k = [\mathbf{0}_{N_S \times n_r} , \tV_S^k] \in \mathbb{C}^{N_S \times T_{RD}},2 \leq k \leq K',
\end{equation}

and the relay sends
\begin{align}
\tX_R^k = [\tI_{n_r}  , \tV_R^k] \in \mathbb{C}^{n_r \times T_{RD}},
\end{align}
where $n_r \leq \min \{N_S, N_D-N_S\}$, $\tV_R^1 \in \mathbb{C}^{N_S \times (T_{RD} - n_r - N_s)}$ and $\tV_R^k \in \mathbb{C}^{N_S \times (T_{RD} - n_r)}$, $k = 2,\dots,K$. 

During the first $T_{RD}$ time slots, the received signal at the destination is
\begin{align}
\nonumber
\tY_D^1 
	  = [\tH_{SD},\tH_{RD}^1] \Bigg[\tI_{(N_S + n_r)} , 
	  \begin{bmatrix}
      \tV_S^1 \\
      \tX_D^1
      \end{bmatrix}
	  \Bigg]
	  + \tW_D^1.
\end{align}

The destination estimates $[\tH_{SD},\tH_{RD}^1]$ during the first $(N_S + n_r)$ time slots and decodes $\tV_S^1$ and $\tV_R^1$. Then every $T_{RD}$ time slots, the received signal at the destination is
\begin{align}
\nonumber
\tY_D^k 
	  = \Bigg[\tH_{RD}^k, ~ [\tH_{SD} ,\tH_{RD}^k]
	  \begin{bmatrix}
      \tV_S^k \\
      \tV_R^k
      \end{bmatrix}
	  \Bigg]
	  + \tW_D^k.
\end{align}

In the following $(K - 1)$ sub-block of length $T_{SD}$, the source-relay channel $\tH_{SR}$ keeps constant and has already been estimated by the relay. Therefore, we copy the transmission strategy in the proof of Theorem~\ref{Thm:TSRinfty}, when the relay knows the channel $\tH_{SR}$. The source and the relay send the signals  \eqref{sig:TH1C2XS1}, \eqref{sig:TH1C2XR1}, \eqref{sig:TH1C2XS2}, and \eqref{sig:TH1C2XR2}. 

The degrees of freedom the relay can decode are $(K - 1)N_S^\ast N_S$. The source-destination link provides total degrees of freedom $KN_S(K'T_{RD} - K'n_r -N_S)$. The relay-destination link can provide total degrees of freedom $Kn_r(K'T_{RD} - K'n_r -N_S)$ . Adding it up with the degrees of freedom the source-destination link achieves and optimizing the number of the relay transmit antennas, the achievable degrees of freedom in \eqref{eq:TH2D3} are obtained.
This completes the proof.
\end{proof}

\color{black}
\begin{figure}
\begin{minipage}[t]{3.5in}
\centering
\includegraphics[width=\Figwidth]{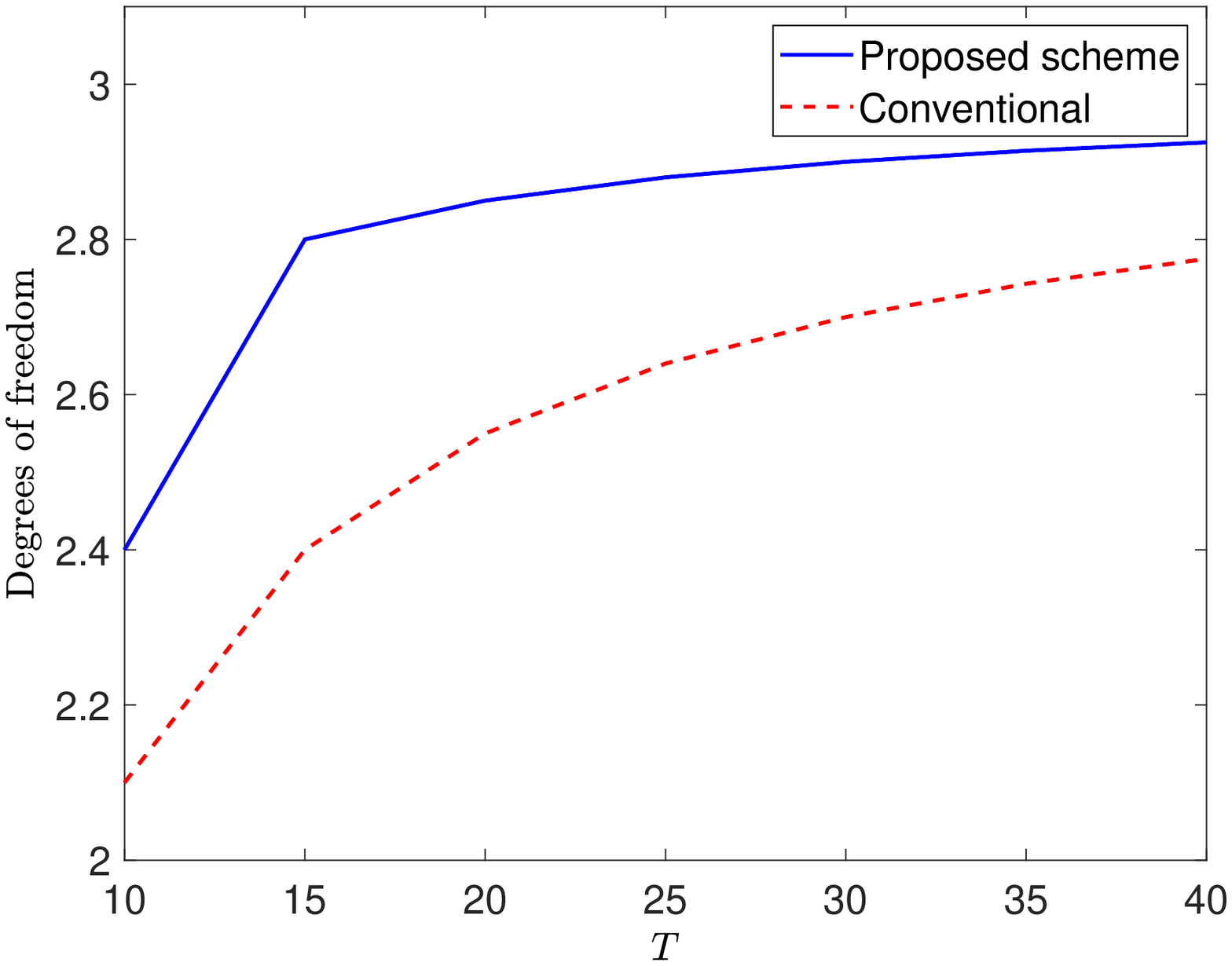}
\caption{DoF for $T_{SR} = \infty$, $T_{RD} = T_{SD} = T$.}
\label{fig:DoF1}
\end{minipage}
\begin{minipage}[t]{3.5in}
\centering
\includegraphics[width=\Figwidth]{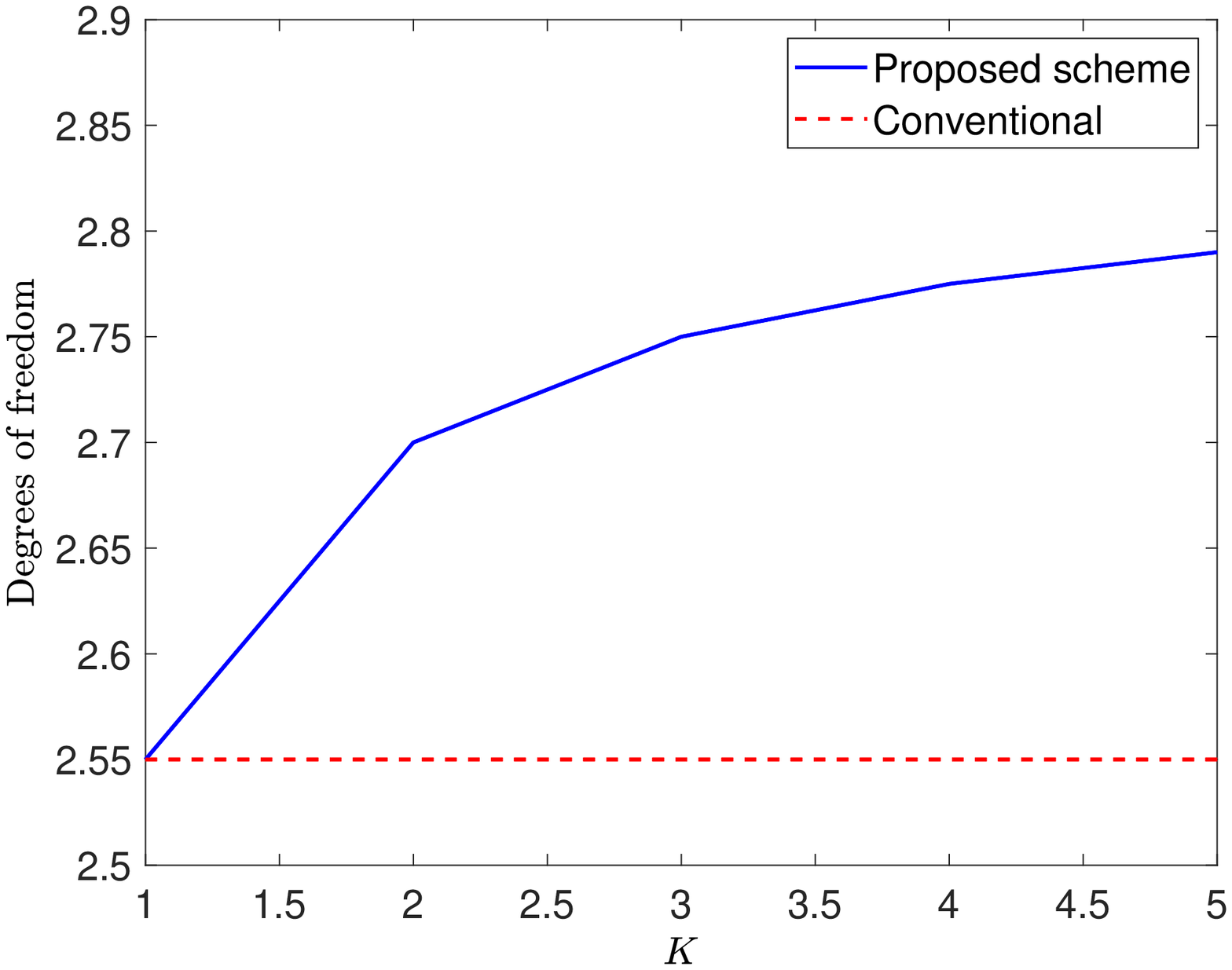}
\caption{DoF for $T_{SR} = KT_{RD} = KT_{SD} = KT$, $T=10$.}
\label{fig:DoF2}
\end{minipage}
\end{figure}
Figure~\ref{fig:DoF1} compares the performance of the proposed scheme with a conventional transmission strategy which designs signals according to the shortest coherence time,  demonstrating the gains in degrees of freedom.  The antenna configuration is $N_S = N_R =3$ and $N_D = 5$. The coherence intervals are $T_{SR} = \infty$, $T_{RD} = T_{SD} = T$. The proposed scheme has a significant gain in degrees of freedom over the conventional transmission. Figure~\ref{fig:DoF2} considers the case where $T_{SR} = KT_{RD} = KT_{SD} = KT $, $T = 10$ for different $K$. When $K = 1$, i.e., all links have identical coherence times, and there is no degrees of freedom gain to be obtained; when $K$ grows, the gain achieved by the proposed scheme increases.

\begin{remark}
The special case $T_{SD}=\infty$, and either $T_{SR}<\infty$ or $T_{RD}<\infty$ (or both), 
$d = \min\{N_S, N_D\}$. This is a corollary of Proposition~\ref{Thm:Equal} for $T_{SD}=T_{SR}=T_{RD}=\infty$, using the fact that reduction of any coherence time does not increase the DoF, and recognizing that DoF $\min\{N_S, N_D\}$ is achievable in these cases by deactivating the relay.
\end{remark}

\section{Achievable DoF with Relay Scheduling}
\label{sec:scheduling}
In this section, a new scheme combining product superposition and relay scheduling is introduced. \newstuff{The relay scheduling is implemented to make a balance between the time slots for estimation in pilot-based scheme and the degrees of freedom of message the relay can decode.} The following theorem highlights the main result of this section. For convenience and compact expression of results, we define:
\begin{align*}
d_1 & \triangleq N_S(T_{SD} - n_r - N_S),\\ 
d_2 & \triangleq n_r(T_{SD} - n_r - N_S),\\ 
d_3 & \triangleq N_S \min\{N_S,N_R\}  , 
\end{align*}

\begin{theorem}
\label{Thm:mod}
In a relay channel with link coherence times $T_{SR} = \infty$, $T_{SD} = T_{RD} = T$ and antenna configuration $N_S  < N_D$, under aligned coherence blocks, 
\begin{itemize}
    \item If $d_2 \leq d_3$, the degrees of freedom $d = \frac{1}{T_{SD}} \max_{n_r} (d_1 + d_2)$ are achievable. 
    \item If $d_2 > d_3$, the following degrees of freedom are achievable,
\end{itemize}
\begin{align}
\label{dof:mod}
d & = \frac{1}{T_{SD}} \max_{n_r} (\frac{d_2 - d_3}{d_2} N_S(T_{SD} - N_S) + \frac{d_3}{d_2}(d_1 + d_2)).
\end{align}
\end{theorem}

\begin{proof}
If $d_2 \leq d_3$, the achievable degrees of freedom follow Theorem~\ref{Thm:TSRinfty}. When $d_2 > d_3$, the transmit scheme with relay scheduling has two phases, each of them lasting an integer multiple of the coherence interval $T$. In both phases, product superposition is used at the source, but the relay action is different in the two phases. We transmit for $d_2-d_3$ coherence intervals in Phase~1, followed by transmitting $d_3$ coherence intervals in Phase~2.

During Phase~1, the relay transmission is deactivated, but the source continues to transmit via product superposition. In this phase, in each coherence interval of length $T_{SD}$, the source delivers to the destination data rates corresponding to its point-to-point degrees of freedom bound, which is $N_S(T_{SD} - N_s)$, while delivering {\em additional} data to the relay with degrees of freedom $d_3$. We transmit in Phase~1 for $d_2 - d_3$ coherence intervals, therefore, the normalized (per-symbol) average degrees of freedom contribution of this phase is $\frac{d_2 - d_3}{d_2} \frac{1}{T_{SD}} N_S(T_{SD} - N_S)$.

During Phase 2, the relay is activated and the source sends the product superposition signal
\begin{equation}
\tX_S = \tU [\tI_{N_S} , \mathbf{0}_{N_S \times n_r} , \tV_S],
\end{equation}
where $n_r \leq \min \{N_S, N_D-N_S\}$, $\tU \in \mathbb{C}^{N_S \times N_S}$ and $\tV_S \in \mathbb{C}^{N_S \times (T_{SD} - n_r - N_S)}$.

The relay knows $\tH_{SR}$ and decodes $\tU$. The relay uses $n_r$ antennas for transmission, sending
\begin{align}
\tX_R = [\mathbf{0}_{n_r \times N_S} , \tI_{n_r}  , \tV_R] \in \mathbb{C}^{n_r \times T_{SD}},
\end{align}
where $\tV_R \in \mathbb{C}^{n_r \times (T_{SD} - n_r - N_S)}$.

The destination estimates the equivalent channel $[\tH_{SD}\tU,~\tH_{RD}]$ during the first $(N_S + n_r)$ time slots and then decodes its messages. Destination receives: $\tV_S$ from the source and $\tV_R$ from the relay, providing degrees of freedom $d_1$ and $d_2$, respectively. Phase 2 consists of $d_3$ coherence intervals; further, recall that the relay has stored data available from Phase~1 in addition to the data it is receiving during Phase~2. Therefore, the relay can send data with degrees of freedom $d_2$ to the destination. Hence during phase~2, the normalized per-symbol degrees of freedom are $\frac{1}{T_{SD}} \frac{d_3}{d_2}(d_1 + d_2)$. 

Adding the degrees of freedom achieved in Phase~1 and Phase~2 and optimizing the number of relay transmit antennas to be activated produces~\eqref{dof:mod}. This completes the proof.
\end{proof}

\begin{remark}
For comparison, we also mention the degrees of freedom {\em without} relay scheduling. For a relay with the following setup $T_{SR} = \infty$, $T_{SD} = T_{SR} = T$ and $N_S  < N_D$. From Theorem~\ref{Thm:TSRinfty}, the following degrees of freedom are achievable:
\begin{align*}
d = \frac{1}{T_{SD}} \max_{n_r} \min \{d_1 + d_2,d_1 + d_3\}.
\end{align*}
\end{remark}

\begin{figure*}
\centering
\includegraphics[width=4.5in]{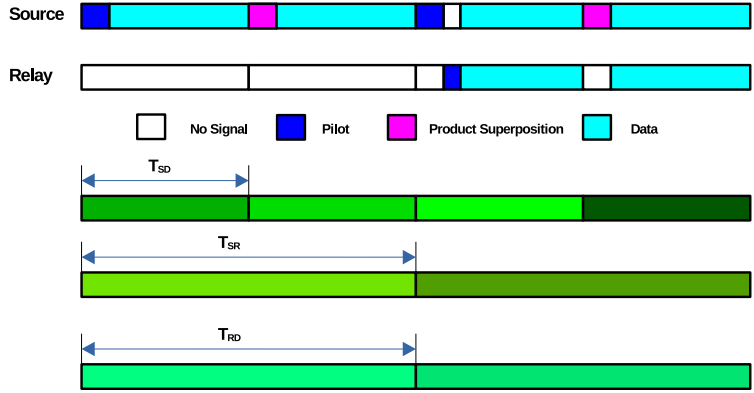}
\caption{Signal Structure with Relay Scheduling.}
\label{fig:Sig_Mod}
\end{figure*}

Figure~\ref{fig:Sig_Mod} shows the signaling structure of the proposed scheme combining product superposition and relay scheduling. Figure~\ref{fig:DoF_mod} shows the comparison between the achievable degrees of freedom of product superposition alone and with relay scheduling when $N_S = 3, N_D = 5$ for different $T$. 
\begin{figure}
    \centering
    \includegraphics[width=\Figwidth]{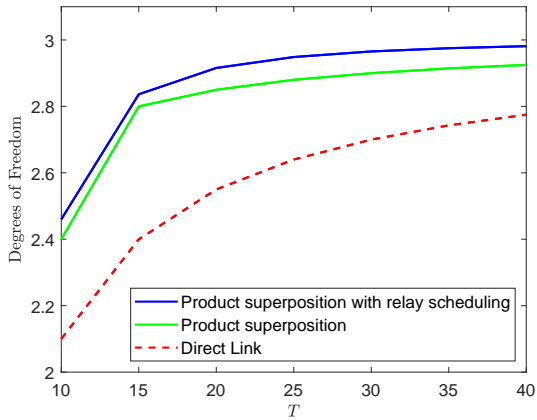}
    \caption{Achievable DoF in Theorem~\ref{Thm:mod}.}
    \label{fig:DoF_mod}
\end{figure}

\section{General Coherence Times}
\label{sec:general}
\subsection{Staggered Coherence Blocks}
We now consider the scenario when the coherence blocks are not perfectly aligned. To build intuition and motivation for the proposed approach, we begin with an unaligned counterpart to the toy example in Section~\ref{sec:toy_unequal}. Then we generalize the result to arbitrary coherence times.

The unaligned toy example is as follows: the source and relay are equipped with 2 antennas and the destination is equipped with 3 antennas. The coherence times of the three links are as follows: $T_{SR} = \infty$, i.e., the source-relay channel is static, therefore the cost of training over this link is amortized over a large number of samples and we can assume the relay knows $\tH_{SR}$. Furthermore we assume $T_{SD} = T_{RD} = 8$. The coherence blocks of the channel $\tH_{RD}$ starts from the 5th time slot of $\tH_{SR}$.

The source uses product superposition, sending
\begin{equation}
\tX_S = \tU [\tI_{2} , \mathbf{0}_{2 \times 1} , \tV_S],
\end{equation}
where $\tU \in \mathbb{C}^{2 \times 2}$ and $\tV_S \in \mathbb{C}^{2 \times 5}$.

At the relay, the received signal is
\begin{align}
\tY_R \twocolAlignMarker = \tH_{SR} \tX_S + \tW_{R} \twocolnewline
      \twocolAlignMarker = \tH_{SR} \tU [\tI_{2} , \mathbf{0}_{2 \times 1} , \tV_S] + \tW_{R}.
\end{align}

The received signal at the first two time slots is
\begin{align}
\tY_R^{\prime}  =  \tH_{SR} \tU  +\tW_{R}^{\prime}.
\end{align}
The relay knows $\tH_{SR}$ and decodes $\tU$. Assume the signal decoded by the relay in the previous block is $\tU^{\prime}$ and the two rows of $\tU^{\prime}$ are $\tu_1^{\prime}, \tu_2^{\prime} \in \mathbb{C}^{1 \times 2}$.

The relay uses one antenna for transmission and sends
\begin{align}
\tX_R = [\mathbf{0}_{1 \times 2} ,1  , \tu_1^{\prime} , \tu_2^{\prime} , 0] \in \mathbb{C}^{1 \times 8}.
\end{align}

Now in one coherence block of $\tH_{SD}$, because of the unaligned blocks of $\tH_{RD}$, the received signal at the destination will experience two realizations of $\tH_{RD}$ in the first 4 time slots,
\begin{align}
\tY_D 
	   = [\tH_{SD}\tU,\tH_{RD1}] \Bigg[\tI_{3} ,
      \begin{bmatrix}
      \tV_S^1 \\
      \tu_1^{\prime}(1)
      \end{bmatrix}
	  \Bigg]
	  + \tW_D.
\end{align}
The destination estimates the equivalent channel $[\tH_{SD} \tU,\tH_{RD1}]$ in the first three time slots and decodes $\tV_S, \tu_1^{\prime}(1)$. 

In the next 4 time slots, the received signal is:
\begin{align}
\tY_D 
	  & = [\tH_{SD} \tU ,\tH_{RD2}]  
	  \begin{bmatrix}
	 \tV_S \\
	  \tu_1^{\prime}(2) , \tu_2^{\prime} , 0
	  \end{bmatrix}
	  + \tW_D.
\end{align}
The first part of the equivalent channel $\tH_{SD} \tU$ is already estimated. The second part $\tH_{RD2}$ will be estimated in the next transmit block. Therefore, the destination decodes $\tV_S^2, \tu_1^{\prime}(2)$ and $\tu_2^{\prime}$. This shows that when the coherence blocks from the source and relay to the destination are unaligned, the proposed scheme can still be used. The destination achieves the same degrees of freedom $d = (2 \times 5 + 2 \times 1 \times 2)/8 = 1.75$ when the coherence blocks are aligned and with the same coherence times. Recall that for a conventional technique that trains all links according to the shortest coherence interval, the degrees of freedom are $d^{\prime} = 2 \times (8 - 2)/8 = 1.5$. 

A similar reasoning can be used to verify that when the coherence blocks from the source to the relay and the source to destination are unaligned, the offset of these coherence blocks will not affect the achievability of our proposed scheme.

\subsection{Arbitrary Coherence Times}
The following theorem states the achievable degrees of freedom with arbitrary coherence times and Figure~\ref{fig:Sig_Stagger} illustrates the signaling structure for the achievable scheme.
\begin{theorem}
\label{thm:DoF_anyTc}
In a relay channel with link coherence times satisfying $T_{SR} > T_{SD}$, $T_{RD} > T_{SD}$.  and antenna configuration $N_S, N_R < N_D$, the following degrees of freedom are achievable:
\begin{align}
d  = & \frac{1}{T_{SR}T_{SD}T_{RD}} \nonumber\\
& \times \max_{n_r} \{N_S(T_{SR}T_{SD}T_{RD} - N_ST_{SR}T_{RD} - n_rT_{SR}T_{SD})  \nonumber\\ 
 & + \min \{N_S^{\ast}N_S(T_{SR}T_{RD} - T_{SD}T_{RD}), \nonumber\\
 & \quad \quad \quad n_r(T_{SR}T_{SD}T_{RD} - N_ST_{SR}T_{RD} - n_rT_{SR}T_{SD}) \}\}.
\end{align}
where  $N_S^{\ast}  \triangleq \min\{N_S,N_R\}$.
\end{theorem}
\begin{proof}
Design the pilot-based achievable scheme in the following manner:
\begin{itemize}
\item On the multiple-access side, pilots sent from the relay and the source will be allocated in different time slots, such that they will not interfere with each other. In addition, during these time slots no data is sent, avoiding pilot contamination.
\item On the broadcast side, the source-relay link needs fewer pilots than the source-destination. Thus, product superposition enables transmission of additional data to the relay. 
\end{itemize}

\begin{figure*}
\centering
\includegraphics[width=4.5in]{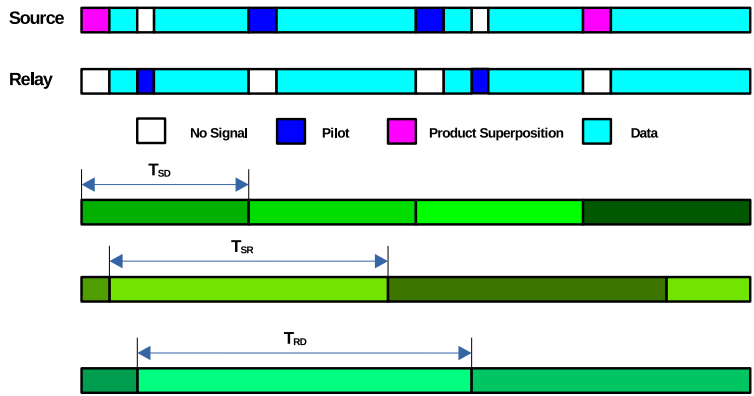}
\caption{Signaling Structure for Arbitrary Coherence Times.}
\label{fig:Sig_Stagger}
\end{figure*}

In the following, we consider a super-interval of length $T_{SR}T_{RD}T_{SD}$, after which the coherence intervals will come back to their original alignment.  The achievable degrees of freedom are calculated as follows:

In each source-destination coherence interval $T_{SD}$, $N_S$ pilot symbols are transmitted. We call the pilot symbols in each coherence block a {\em pilot sequence}.

Therefore, for source-destination link, we repeat the length-$N_S$ pilot sequence $T_{SR}T_{RD}$ times over the length-$T_{SR}T_{RD}T_{SD}$ super-interval. Having coherence time $T_{SR}T_{RD}$, the relay needs $T_{SD} T_{RD}$ pilot sequences. Hence, product superposition can be applied during $T_{SR}T_{RD} - T_{SD}T_{RD}$ pilot sequences of length $N_S$ to send data to the relay. Data with $N_S^{\ast}$ degrees of freedom per symbol can be sent.

Over each super-interval, the relay-destination link needs $T_{SR}T_{SD}$ pilot sequences of length $n_r$. The pilots slots will be non-overlapping with pilots transmitted from the source terminal.

In each super-interval, the source and the relay each have $(T_{SR}T_{SD}T_{RD} - N_ST_{SR}T_{RD} - n_rT_{SR}T_{SD})$ time slots available for sending data. The source has $N_S$ degrees of freedom available per transmission, and the relay $n_r$ degrees of freedom per transmission.

The relay can decode at most $N_S^{\ast}N_S(T_{SR}T_{RD} - T_{SD}T_{RD})$ degrees of freedom, therefore, it provides $\min \{N_S^{\ast}N_S(T_{SR}T_{RD} - T_{SD}T_{RD}), n_r(T_{SR}T_{SD}T_{RD} - N_ST_{SR}T_{RD} - n_rT_{SR}T_{SD}) \}$ degrees of freedom, the minimum of the degrees of freedom the relay can receive and can transmit.

We can now add the degrees of freedom by the source transmission (subject to relay constraints) with the degrees of freedom provided by the relay transmission, and optimize the number of relay antennas to be activated. This concludes the proof.
\end{proof}

The following corollary shows the achievable degrees of freedom when using relay scheduling  with arbitrary coherence times.
\begin{corollary}
Define the following notation:
\begin{align*}
d_1 & \triangleq N_S(1 - \frac{N_S}{T_{SD}} - \frac{n_r}{T_{RD}}),\\ 
d_2 & \triangleq n_r(1 - \frac{N_S}{T_{SD}} - \frac{n_r}{T_{RD}}),\\ 
d_3 & \triangleq N_S^\ast N_S(\frac{1}{T_{SD}} - \frac{1}{T_{RD}}).
\end{align*}

In a relay with coherence diversity, 
\begin{itemize}
    \item If $d_2 \leq d_3$, the degrees of freedom $d =  \max_{n_r} (d_1 + d_2)$ are achievable. 
    \item If $d_2 > d_3$, the following degrees of freedom are achievable. 
\end{itemize}
\begin{align}
\label{dof:anyT_sch}
d & = \max_{n_r} (\frac{d_2 - d_3}{d_2} N_S(1 - \frac{N_S}{T_{SD}}) + \frac{d_3}{d_2}(d_1 + d_2)),
\end{align}
\end{corollary}

\begin{proof}
If $d_2 \leq d_3$, the achievable degrees of freedom follows Theorem~\ref{thm:DoF_anyTc}. When $d_2 > d_3$. the transmit scheme with relay scheduling has two phases. In both phases, product superposition is used at the source, but the relay action is different in the two phases, as described in the sequel. We propose to transmit for $d_2-d_3$ coherence intervals in Phase~1, followed by transmitting $d_3$ coherence intervals in Phase~2.

During Phase 1, the relay transmission is deactivated but the source continues to transmit via product superposition. In this phase, the source delivers to the destination data rates corresponding to its point-to-point degrees of freedom bound, following the result in \cite{Fadel_disparity} which is $N_S(T_{SD} - N_S)$, while delivering {\em additional} data to the relay with degrees of freedom $d_3$. We transmit in Phase~1 for $d_2 - d_3$ coherence intervals, therefore, the normalized (per-symbol) average degrees of freedom contribution of this phase is $\frac{d_2 - d_3}{d_2} \frac{1}{T_{SD}} N_S(T_{SD} - N_S)$.

During Phase 2, following the strategy from the proof of Theorem~\ref{Thm:mod}, the relay has stored data available from Phase~1 in addition to the data it is receiving in Phase~2. Therefore, the relay can send data with degrees of freedom $d_2$ to the destination. Hence during phase~2, the normalized per-symbol degrees of freedom are $\frac{d_3}{d_2}(d_1 + d_2)$. 

Adding the degrees of freedom achieved in Phase~1 and Phase~2 and optimizing the number of relay transmit antennas to be activated produces~\eqref{dof:anyT_sch}. This completes the proof.
\end{proof}

\section{Multiple Relays in Parallel}
\label{sec:multiple}
\begin{figure}
\centering
\includegraphics[width=\Figwidth]{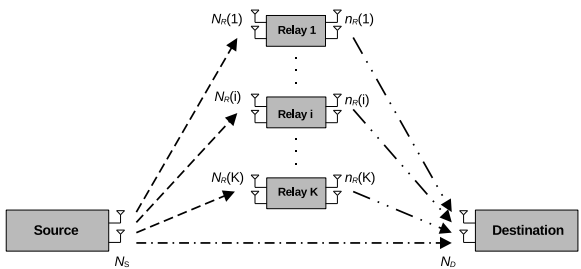}
\caption{MIMO Channel with Multiple Parallel Relays.}
\label{fig:RC_Strc_Mul}
\end{figure}

This section studies the MIMO relay channel with $K$ full-duplex relays, under coherence diversity . The source and destination are equipped with $N_S$ and $N_D$ antennas, respectively. Relay~$k$ has $N_R(k)$ receive antennas and uses $n_R(k) \leq N_R(k)$ antennas for transmission. Figure~\ref{fig:RC_Strc_Mul} shows the structure of the system. The received signals at the relays and the destination are:
\begin{align}
\label{Eq:sys_k1}
\ty_R(k) & = \tH_{SR}(k) \tx_S + \tw_{R}(k), \qquad k=1,\ldots, K\\
\label{Eq:sys_k2}
\ty_D & = \tH_{SD} \tx_S + \sum_{k=1,\dots,K}\tH_{RD}(k) \tx_R(k) + \tw_{D},
\end{align}
where $\tx_S$ and $\tx_R(k)$ are signals transmitted from the source and Relay~$k$. $\tw_R$ and $\tw_D$ are i.i.d. zero-mean Gaussian noise and $\tH_{SR}(k)$, $\tH_{RD}(k)$ and $\tH_{SD}$ are channel gain matrices, whose entries are i.i.d.\ Gaussian. We assume there is no free CSI at the destination and no CSIT at the source or relay. In the parallel relay geometry, there are no inter-relay links \newstuff{and there is no cooperation between these relays}. Denote the coherence time of the link between the source and Relay~$k$ as $T_{SR}(k)$ and the coherence time of the link between Relay~$k$ and the destination as $T_{RD}(k)$. \newstuff{The two parallel relay scenario will first be presented to show implementation of product superposition in multiple links.}

\subsection{Achievable DoF for Two Parallel Relays}
Consider the following channel with two parallel relays.  $T_{SR}(2) = K_2 T_{SR}(1) = K_2 K_1 T_{SD} = K_2 K_1 T$ and the destination knows the channel state of $\tH_{RD}(1)$ and $\tH_{RD}(2)$, i.e., $T_{RD}(1) = T_{RD}(2) = \infty$. Denote $N_S^{\ast}(i) = \min\{N_S,N_R(i)\}$. If Relay~1 or Relay~2 is activated alone, the achievable degrees of freedom are
\begin{align}
d_i  \twocolAlignMarker =\max_{n_R(i)} \Big\{N_S(1 - \frac{N_S}{T})  \twocolnewline
\twocolAlignMarker + \min \big\{(1 - \frac{1}{K_i}) \frac{N_S^{\ast}(i)N_S}{T},n_R(i)(1 - \frac{N_S}{T})  \big\}\Big\}.
\end{align}
When Relay~1 and Relay~2 are both activated, consider a transmission interval of length $K_2K_1T$. During each coherence interval of length $K_2K_1T$, 
Relay~1 and Relay~2 send the messages they decoded in the previous interval. The transmitted signal from Relay~1 and 2 over each sub-interval of length $T$ has the following structure and is repeated $K_2K_1$ times:
\begin{equation}
\tX_R(i) = [\ZeroMat_{n_R(i) \times N_S},\tV_{Ri}], \quad i = 1,2.
\end{equation}
During the first coherence interval of length $K_1T$, in the first sub-interval of length $T$, the source sends $\tX_S = [\tI_{N_S}, \tX_D]$. Relay~1 and Relay~2 estimate their channel. The signal at the destination is
\begin{equation}
\begin{aligned}
\tY_D & = [\tH_{SD}, \tH_{RD}(1), \tH_{RD}(2)]  
	  \begin{bmatrix}
	  \tI_{N_S} , \tV_D\\
	  \mathbf{0}_{n_R(1) \times N_S}, \tV_{R1} \\
	  \mathbf{0}_{n_R(2) \times N_S}, \tV_{R2}
	  \end{bmatrix}
	  + \tW_D \\
	  & = \Bigg[\tH_{SD},[\tH_{SD}, \tH_{RD}(1), \tH_{RD}(2)]
	  \begin{bmatrix}
\tV_D \\
\tV_{R1} \\
\tV_{R2}
\end{bmatrix}
	  \Bigg] + \tW_D.
\end{aligned}
\end{equation}
 The destination estimates $\tH_{SD}$ and decodes the messages in $\tV_D,\tV_{R1}$ and $\tV_{R2}$, which respectively provide degrees of freedom $N_S,n_R(1)$ and $n_R(2)$ per symbol over this interval of length $(T - N_S)$. 

In the remaining $K_1 - 1$ intervals of length $T$, the source sends the signal
\begin{equation}
\tX_S = \tU_1^i[\tI_{N_S} , \tV_D^i], \quad i = 1,2,\dots,K_1 - 1,
\end{equation}
where $\tU_1^i \in \mathbb{C}^{N_S \times N_S}$. Relay~1 has already estimated its channel in the first interval of length $T$. It can decode  $\tU_1^i$, achieving degrees of freedom $N_S^{\ast}(1)N_S$. The total degrees of freedom Relay~1 can decode are $(K_1 - 1)N_S^{\ast}(1)N_S$. The received signal at the destination is
\begin{align}
\tY_D \twocolAlignMarker = \bigg[\tH_{SD}\tU_1^i,[\tH_{SD}\tU_1^i, \twocolnewline
\twocolAlignMarker \quad \quad \quad \tH_{RD}(1), \tH_{RD}(2)]
\begin{bmatrix}
\tV_D \\
\tV_{R1} \\
\tV_{R2}
\end{bmatrix}
\bigg] + \tW_D.
\end{align}
 The destination estimates $\tH_{SD}\tU_1^i$ and decodes  $\tV_D$, $\tV_{R1}$, and $\tV_{R2}$.

During each of the remaining $K_2 - 1$ coherence intervals of length $K_1 T$, the transmitter sends a signal with the same structure as the first sub-interval of length $T$, multiplying it from the left by $\tU_2^j$, which contains the message for Relay~2. During each interval of length $K_1T$, the transmitted signal from the source has the following structure
\begin{align}
\tX_S \twocolAlignMarker = \tU_2^j\bigg[[\tI_{N_S} , \tX_D^1], \tU_1^1[\tI_{N_S} , \tX_D^2],\twocolnewline
\twocolAlignMarker \quad \quad \quad \quad \tU_1^2[\tI_{N_S}  \tX_D^3],\dots, 
\tU_1^{(K_1 - 1)}[\tI_{N_S} , \tX_D^{K_1}]\bigg].
\end{align}
During these $K_2 - 1$ coherence intervals with length $K_1 T$, the channel $\tH_{SR}(2)$ remains the same as in the first sub-interval of length $K_1 T$. Therefore, in each interval of length $K_1 T$, \textcolor{black}{Relay~2 can achieve degrees of freedom  $N_S^{\ast}(2)N_S$.} The total degrees of freedom Relay~2  can decode are $(K_2 - 1)N_S^{\ast}(2)N_S$ over coherence interval of length $K_2K_1T$.

At Relay~1, the first $N_S$ symbols received during the first sub-interval of length $K_1T$ are
\begin{equation}
    \tY_R(1) = \tH_{SR}^j\tU_2^j + \tW_R(1).
\end{equation}
The first $N_S$ symbols during the remaining sub-interval of length $K_1T$ received at Relay~1 are
\begin{equation}
    \tY_R(1) = \tH_{SR}^j\tU_2^j\tU_1^i + \tW_R(1), \  i = 1,\dots,K_1 - 1.
\end{equation}
Relay~1 first estimates its equivalent channel 
\begin{equation}
\tilde{\tH}_{SR}^j(1) =  \tH_{SR}^j(1) \tU_2^j,
\end{equation}
and decodes  $\tX_R^i(1)$, which provides degrees of freedom  $N_S^{\ast}(1)N_S$. The total degrees of freedom Relay~1 can decode are $(K_2 - 1)(K_1 - 1)N_S^{\ast}(1)N_S$.

At the destination, the received signal during the first sub-interval of length $K_1T$ is
\begin{align}
\tY_D 
\twocolAlignMarker = \Bigg[\tH_{SD}\tU_2^j,[\tH_{SD}\tU_2^j, \twocolnewline
\twocolAlignMarker \quad \quad \quad \tH_{RD}(1),\tH_{RD}(2)]
\begin{bmatrix}
\tV_D \\
\tV_{R1} \\
\tV_{R2}
\end{bmatrix}
\Bigg] + \tW_D,
\end{align}
 and the received signals during each remaining sub-intervals of length $K_1T$ are 
\begin{align}
\tY_D
\twocolAlignMarker = \Bigg[\tH_{SD}\tU_2^j\tU_1^i,[\tH_{SD}\tU_2^j\tU_1^i, \twocolnewline
\twocolAlignMarker \tH_{RD}(1),\tH_{RD}(2)]
\begin{bmatrix}
\tV_D \\
\tV_{R1} \\
\tV_{R2}
\end{bmatrix}
\Bigg] + \tW_D,
\end{align}
where $i = 1, \dots, K_1 - 1$, $j = 1, \dots,K_2 - 1$. The destination estimates the equivalent channel $\tH_{SD}\tU_2^j$, $\tH_{SD}\tU_2^j\tU_1^i$, and decodes  $\tX_D,\tV_{R1}$ and $\tV_{R2}$, which respectively provide degrees of freedom $N_S, n_R(1), n_R(2)$ per symbol over each time interval of length $K_1 T$. 

During each interval of length $K_2K_1T$, the source-destination link can always provide degrees of freedom $N_S(1 - \frac{N_S}{T})$ per symbol. The maximum degrees of freedom decoded at Relay~1 are $(K_1 - 1)N_S^{\ast}(1)N_S + (K_2 - 1)(K_1 - 1)N_S^{\ast}(1)N_S = K_2(K_1 - 1)N_S^{\ast}(1)N_S$. The degrees of freedom decoded at Relay~2 are $(K_2 - 1)N_S^{\ast}(2)N_S$. During each interval of length $K_2K_1T$, the number of time slots available to relays for sending data is $K_2K_1(T - N_S)$. The degrees of freedom the relays can provide via the relay-destination links are $n_R(i)K_2K_1(T - N_S), i = 1,2$. Noting that the emitted data by the relays is limited by what they can decode, we add the degrees of freedom by the two relays, normalize it per symbol, and optimize the number of transmit antennas activated at the relays. The following degrees of freedom are achievable
\begin{align}
d  = & \max_{n_R(i)} \bigg\{N_S(1 - \frac{N_S}{T}) \nonumber\\
& + \min\big\{(1 - \frac{1}{K_1})\frac{N_S^{\ast}(1)N_S}{T}, 
 n_R(1)\frac{T - N_S}{T}\big\} \nonumber\\
 & +  \min\big\{\frac{K_2 - 1}{K_1K_2}\frac{N_S^{\ast}(2)N_S}{T},n_R(2)\frac{T - N_S}{T} \big\}\bigg\}.
\end{align}

\subsection{Achievable DoF for $K$ Parallel Relays}

We now extend the ideas and techniques that were developed in the two-relay framework to the $K$-relay case. In the interest of economy of expression, the parts that are similar to the earlier discussions are condensed or omitted.

Denote with $\mathbf{T}_{SR}$ and $\mathbf{T}_{RD}$ the size-$K$ vectors containing, respectively, source-relay and relay-destination coherence times, and $\mathbf{N}_R,\mathbf{n}_R$ the number of receive and activated transmit antennas at the relays. Also, we allow a subset $k$ of relays to be used. We denote the coherence times of selected relays with size-$k$ vectors $\mathbf{T}',\mathbf{T}''$ and the number of receive and activated transmit antennas in selected relays with size-$k$ vector $\mathbf{N}',\mathbf{n}'$. The following result shows the achievable degrees of freedom, which is maximized over selected relays and their activated transmit antennas. We define a selection matrix $\mathbf{P}_{k\times K}$ containing $k$ rows of the identity matrix $\tI_{K\times K}$, corresponding to the $k$ indices of the selected relays.

\begin{theorem}
For the multi-relay system \eqref{Eq:sys_k1} and \eqref{Eq:sys_k2}, the following degrees of freedom are achievable:
\begin{alignat}{1}
\label{Eq:dof_k_relay}
& d  = \max_{\mathbf{P},\mathbf{n}',k} \Big\{N_S(1 - \frac{N_S}{T_{SD}} - \sum_{i = 1}^k \frac{n'_i}{T''_i}) \twocolnewline
\twocolAlignMarker + \sum_{i = 1}^k \min \big\{N^{\ast}_i N_S(\frac{1}{T'_{i-1}} - \frac{1}{T'_i}),  n'_i(1 - \frac{N_S}{T_{SD}} - \sum_{j = 1}^k \frac{n'_j}{T''_j})\big\} \Big\},\nonumber\\[0.1in]
&\text{subject to:} \quad  [\mathbf{T}'\; \mathbf{T}'' \;\mathbf{N}'\;\mathbf{n}'] = {\mathbf P} [ \mathbf{T}_{SR}\; \mathbf{T}_{RD} \;\mathbf{N}_R\;\mathbf{n}_R],
\end{alignat}
where $T'_0 \triangleq T_{SD}$, $\mathbf P$ is a selection matrix consisting of $k$ rows of the identity matrix of size $K$, and  $N^{\ast}_i = \min \{N_S, N'_i\}$. 
\end{theorem}

\begin{proof}
The transmit scheme is designed in the same spirit as Theorem~\ref{thm:DoF_anyTc}: On the multiple-access side, pilots sent from the relays and the source are allocated in different time slots; On the broadcast side, product superposition enables transmission of additional data to the relays. 
Throughout this proof, we index only the activated relays, e.g.,  Relay~$i$ refers to $i$-th activated relay. Without loss of generality, $T'_1 \leq T'_2 \leq \dots \leq T'_k$. Define $T_1 \triangleq \prod_{i = 1}^{k}T'_i$ and $T_2 \triangleq \prod_{i = 1}^{k}T''_i$. In the following, we consider a super-interval of length $T_{1}T_{2}T_{SD}$,

During each coherence interval of length $T'_i$, Relay~$i$ needs $T_{SD} T_2 T_1/T'_i$ pilot sequences each of length $N_S$ for channel estimation. Relay~$(i-1)$ needs $T_{SD} T_2 T_1/T'_{i-1}$ pilot sequences each of length $N_S$. Therefore, product superposition can be applied during $(T_{SD}T_2T_1/T'_{i-1} - T_{SD}T_2T_1/T'_{i})$ pilot sequences each of length $N_S$ to send data to Relay~{$i$}, providing $N^{\ast}_i$ degrees of freedom per symbol.

During each coherence interval of length $T_{SD}$ in the source-destination link, $N_S$ pilot symbols are transmitted. In each super-interval (see above) $T_1T_2$ pilot sequences of length $N_S$ are transmitted.

For channel estimation between Relay~$i$ and the destination, during the super-interval of length $T_1T_2T_{SD}$, the destination needs $T_1T_2T_{SD}/T''_{i}$ pilot sequences of length $n'_i$.

Therefore, In each super-interval, the source and relays can use $(T_{SD}T_1T_2 - N_ST_1T_2 - \sum_{i=1}^k\frac{n'_i}{T''_i}T_1T_2T_{SD})$ time slots to send data. The source has $N_S$ degrees of freedom available per transmission, and Relay~{$i$} has $n'_i$ degrees of freedom per transmission.

The decodable degrees of freedom Relay~$i$ are at most $N'_iN_S(T_{SD}T_2T_1/T'_{i-1} - T_{SD}T_2T_1/T'_{i})$. Therefore, the degrees of freedom Relay~$i$ can provide are:
\begin{align*}
    \twocolAlignMarker \min \big\{N'_iN_S(\frac{T_{SD}T_2T_1}{T'_{i-1}}- \frac{T_{SD}T_2T_1}{T'_i}),  \twocolnewline
    \twocolAlignMarker \quad \quad \quad n'_i(T_{SD}T_1T_2 - N_ST_1T_2 - \sum_{i=1 }^k\frac{n'_i}{T''_i}T_1T_2T_{SD})\big\},
\end{align*} 
 the minimum of the degrees of freedom Relay~$i$ can receive and can transmit.

We can now sum the degrees of freedom by the source and the relays and normalize it per symbol. Optimizing the relays to be activated (over $k$ and $\mathbf{P}$) and the number of transmit antennas at the relays $\mathbf{n}'$, the degrees of freedom in \eqref{Eq:dof_k_relay} are achieved. This concludes the proof.
\end{proof}

Figure~\ref{fig:DoF3} and \ref{fig:DoF4} show the achievable degrees of freedom with two parallel relays equipped with $N_S = 3$, $N_D = 6$, $N_R(1) = N_R(2) = 1$ antennas. In Figure~\ref{fig:DoF3}, $T_{SD} = 5$, $T_{RD}(1) = T_{RD}(2) = \infty$, $\frac{T_{SR}(1)}{T_{SR}(2)} = \frac{2}{3}$. In Figure~\ref{fig:DoF4}, $T_{SD} = 5$, $T_{RD}(1) = T_{RD}(2) = \infty$, $T_{SR}(1) = 6$ and different $\frac{T_{SR}(2)}{T_{SR}(1)}$. 

\begin{figure}
\begin{minipage}{3.25in}
\centering
\includegraphics[width=3.25in]{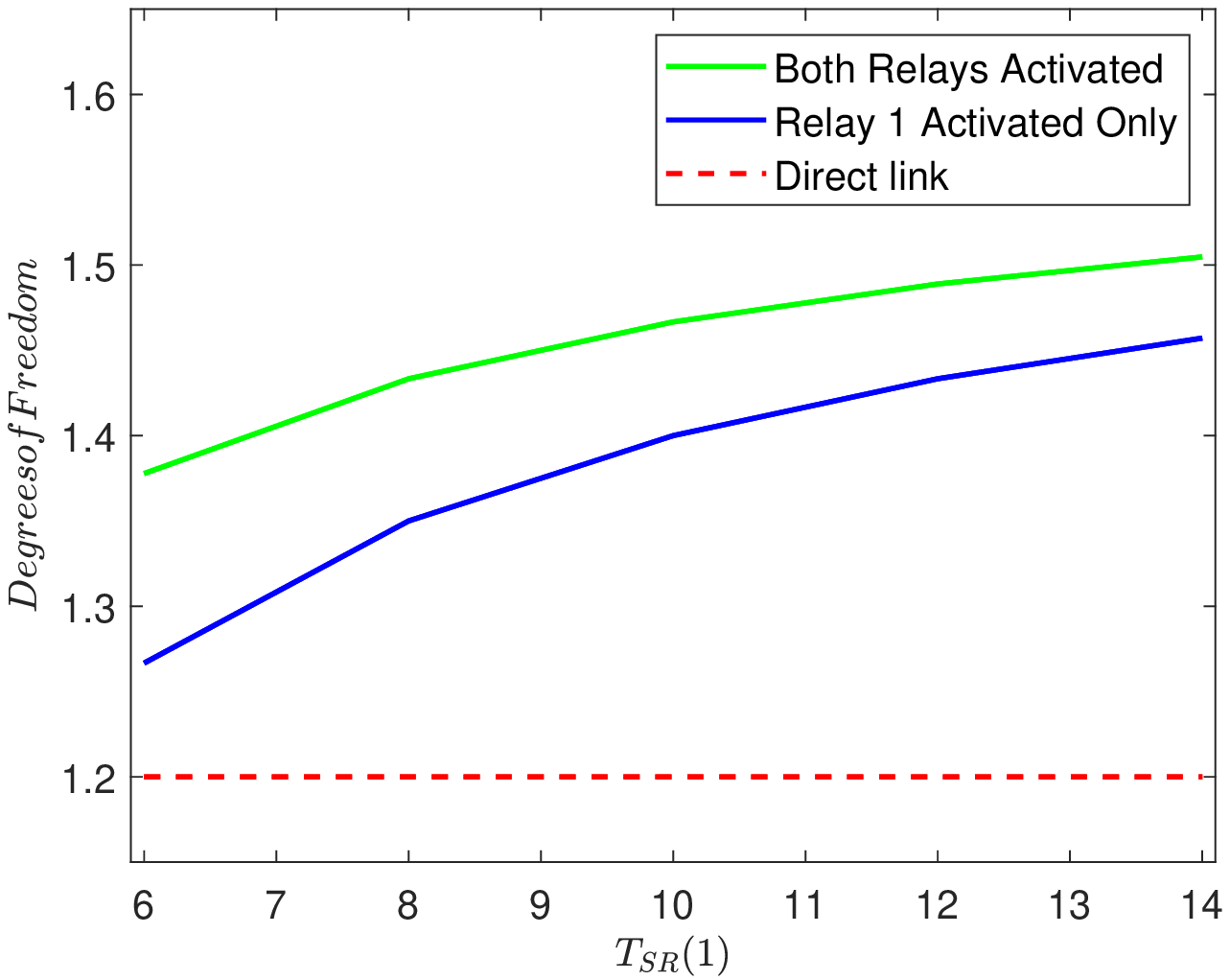}
\caption{DoF subject to $\frac{T_{SR}(1)}{T_{SR}(2)}=\frac{2}{3}$.}
\label{fig:DoF3}
\end{minipage}
\begin{minipage}{3.25in}
\centering
\includegraphics[width=3.25in]{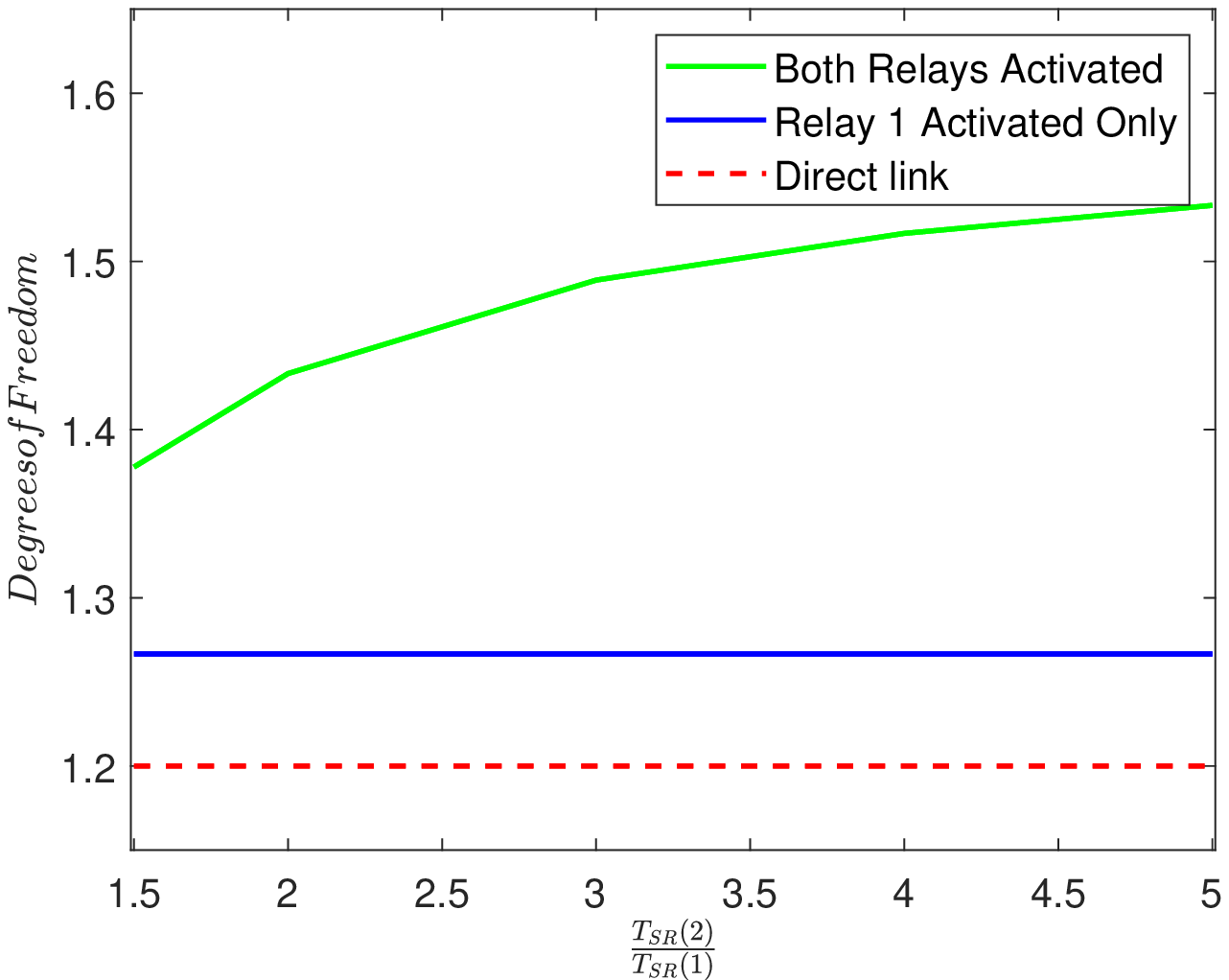}
\caption{DoF subject to $T_{SR}(1) = 6$ with different $\frac{T_{SR}(2)}{T_{SR}(1)}$.}
\label{fig:DoF4}
\end{minipage}
\end{figure}

\section{Conclusion}
\label{sec:conclusion}
This paper studies MIMO relays under coherence diversity. The main contribution of this paper is showing new degrees of freedom gains in the relay channel when the relay links experience unequal coherence times. We  propose and analyze transmission schemes achieving these gains. A key novelty of this paper is to carefully weave product superposition into a signaling structure that controls  the interference by the relay(s) and the source on pilots of source and relay(s), respectively. This is combined with a relay scheduling scheme to balance the relaying gains against the interference cost of relay transmission during source pilots. The proposed techniques are introduced in the case where link coherence times are aligned, and then extended to unaligned coherence intervals. We also analyze multiple parallel relays, and the corresponding achievable degrees of freedom are calculated.

The results of the paper can be extended to the half-duplex case in a straight forward manner: like the full-duplex case, the minimum of the DoF of the broadcast and the DoF of MAC must be calculated (this time subject to half-duplex constraint). Then, the ratio of the two intervals must be optimized to give the best degrees of freedom.

\bibliographystyle{IEEEtran}
\bibliography{IEEEabrv,mybib}

\end{document}